\documentclass[aps,pra,twocolumn,nofootinbib, superscriptaddress,10pt]{revtex4-1}

\usepackage{graphicx}
\usepackage{epstopdf}
\usepackage{amsmath}
\usepackage{amssymb}
\usepackage{mathrsfs}
\usepackage{amsthm}
\usepackage{bm}
\usepackage{url}
\usepackage[T1]{fontenc}
\usepackage{csquotes}
\MakeOuterQuote{"}

\newtheoremstyle{note}
  {\topsep/2}               
  {\topsep/2}               
  {}                      
  {\parindent}            
  {\itshape}              
  {.}                     
  {5pt plus 1pt minus 1pt}
  {}

\theoremstyle{note}
\newtheorem{theorem}{Theorem}
\newtheorem{lemma}{Lemma}

\newtheorem{corollary}{Corollary}
\newtheorem{proposition}{Proposition}

\theoremstyle{definition}

\theoremstyle{remark}
\newtheorem{remark}{Remark}

\def\vec#1{\bm{#1}} 

\newcommand{\tr}{\operatorname{tr}}

\newcommand{\diag}{\operatorname{diag}}

 \newcommand{\rmd}{\mathrm{d}}
 \newcommand{\rme}{\mathrm{e}}
 \newcommand{\rmi}{\mathrm{i}}

 \newcommand{\rmB}{\mathrm{B}}

 \newcommand{\rmT}{\mathrm{T}}

 \newcommand{\caF}{\mathcal{F}}
 \newcommand{\caH}{\mathcal{H}}
 
 \newcommand{\caM}{\mathcal{M}}

\newcommand{\be}{\begin{equation}}
\newcommand{\ee}{\end{equation}}
\newcommand{\ba}{\begin{align}}
\newcommand{\ea}{\end{align}}

\def\<{\langle}  
\def\>{\rangle}  

\newcommand{\dket}[1]{| #1\>\!\>}

\newcommand{\dbra}[1]{\<\!\< #1|}

\newcommand{\dinner}[2]{\<\!\< #1| #2\>\!\>}

\newcommand{\douter}[2]{| #1\>\!\>\<\!\< #2|}

\newcommand{\mse}{\mathcal{E}}

\newcommand{\msb}{\mathcal{E}_{\mathrm{SB}}}

\newcommand{\ml}{\mathrm{ML}}

\def\eqref#1{\textup{(\ref{#1})}}  
\newcommand{\eref}[1]{Eq.~\textup{(\ref{#1})}}

\newcommand{\fref}[1]{Fig.~\ref{#1}}

\newcommand{\sref}[1]{Sec.~\ref{#1}}

\newcommand{\thref}[1]{Theorem~\ref{#1}}
\newcommand{\Thref}[1]{Theorem~\ref{#1}}
\newcommand{\thsref}[1]{Theorems~\ref{#1}}
\newcommand{\Thsref}[1]{Theorems~\ref{#1}}

\newcommand{\lref}[1]{Lemma~\ref{#1}}

\newcommand{\pref}[1]{Proposition~\ref{#1}}
\newcommand{\Pref}[1]{Proposition~\ref{#1}}

\newcommand{\crref}[1]{Corollary~\ref{#1}}
\newcommand{\Crref}[1]{Corollary~\ref{#1}}

\newcommand{\cref}[1]{Conjecture~\ref{#1}}
\newcommand{\Cref}[1]{Conjecture~\ref{#1}}

\newcommand{\rcite}[1]{Ref.~\cite{#1}}
\newcommand{\rscite}[1]{Refs.~\cite{#1}}

\begin{document}
\title{Universally Fisher-Symmetric Informationally Complete Measurements}

\author{Huangjun Zhu}
\email{zhuhuangjun@fudan.edu.cn}
\affiliation{Institute for Theoretical Physics, University of Cologne,  Cologne 50937, Germany}

\affiliation{Department of Physics and Center for Field Theory and Particle Physics, Fudan University, Shanghai 200433, China}

\affiliation{Institute for Nanoelectronic Devices and Quantum Computing, Fudan University, Shanghai 200433, China}

\affiliation{State Key Laboratory of Surface Physics, Fudan University, Shanghai 200433, China}

\affiliation{Collaborative Innovation Center of Advanced Microstructures, Nanjing 210093, China}

\author{Masahito Hayashi}
\affiliation{Graduate School of Mathematics, Nagoya University, Nagoya, 464-8602, Japan}

\affiliation{Centre for Quantum Technologies, National University of Singapore, 3 Science Drive 2, 117542, Singapore}

\begin{abstract}
A quantum measurement is Fisher symmetric if it provides uniform and maximal information on all parameters that characterize the quantum state of interest. Using (complex projective) 2-designs, we  construct measurements on a pair of identically prepared quantum states that are  Fisher symmetric for all pure states. Such measurements are optimal in achieving the minimal statistical error without  adaptive measurements.
We then determine all collective measurements on a pair that are Fisher symmetric for the completely mixed state
and  for all pure states simultaneously.  
For a qubit,  these measurements are  Fisher symmetric for all  states. The minimal optimal measurements are tied to the elusive symmetric informationally complete measurements, which reflects a deep connection between local symmetry and global symmetry. In the study, we derive a fundamental constraint on the Fisher information matrix of any collective measurement on a pair, which offers a useful tool for characterizing the tomographic efficiency of collective measurements.

\end{abstract}

\date{\today}
\maketitle

\emph{Introduction.}---Quantum state tomography is a primitive of various quantum information processing tasks, such as quantum communication and metrology \cite{Hels76book, Hole82book,  Haya05book,BrauC94, PariR04Book,LvovR09,GiovLM11}. Crucial to achieving high tomographic efficiency is a judicial choice of the quantum measurement, which is usually  represented by a positive-operator-valued measure (POVM). A POVM is \emph{informationally complete} (IC) if all states  can be determined uniquely by the measurement statistics.  A symmetric informationally complete POVM (SIC for short) is a special IC POVM that is distinguished by global symmetry between POVM elements \cite{FuchHS17,Zaun11,ReneBSC04}. SICs are optimal on average in minimal or linear quantum state tomography \cite{Scot06,ZhuE11,Zhu14T} and are also interesting for many other reasons, including foundational studies \cite{FuchS13,ApplFSZ17,Zhu16Q,FuchHS17}.

In the paradigm of local tomography,  the quantum state is known to be in the neighborhood of a fiducial state. 
In this case,  the Fisher information matrix is a useful tool for analyzing the performance of a quantum measurement as its inverse sets a lower bound for the mean-square-error (MSE) matrix of any unbiased estimator \cite{Rao02book}. A quantum measurement is \emph{Fisher symmetric} if it provides uniform and maximal information on all parameters that characterize the quantum state of interest~\cite{LiFGK16}. Here uniformity means that the Fisher information matrix is proportional to the quantum Fisher information matrix \cite{Hels76book, Hole82book,Haya05book,BrauC94}, and  maximality means that no other measurement can provide more information \cite{GillM00}. Such measurements, if they exist, are as efficient as possible for estimating all parameters of interest.

For pure states, Li et al. \cite{LiFGK16} offered a  method for constructing  measurements that are Fisher symmetric for  any arbitrary, but fixed state. The covariant measurement composed of all pure states weighted by the Haar measure is simultaneously Fisher symmetric for all pure states \cite{Haya98,Zhu12the,Zhu14IOC}, but is not  realistic to implement.
For mixed states, Fisher-symmetric  measurements in general cannot exist except for the completely mixed state \cite{LiFGK16}. Even the covariant measurement is far from being Fisher symmetric for a generic state, which means mixed states cannot be estimated efficiently by  fixed measurements if  infidelity or  Bures distance is the figure of merit \cite{MahlRDF13,Zhu12the,Zhu14IOC,HouZXL16}.

In this paper, we show that many limitations mentioned above can be overcome if we can measure  a pair of identically prepared quantum states together. Such collective measurements are known to provide more information than separable measurements \cite{PereW91,MassP95,BagaBGM06S,Zhu12the,VidrDGJ14} and  are tied to  a number of nonclassical phenomena \cite{BennDFM99,GisiP99}. 
Using complex projective 2-designs \cite{Zaun11,ReneBSC04,Scot06}, we  construct quantum measurements  that are universally Fisher symmetric for all pure states and that have no more than $4d^2$ outcomes for a $d$-level system. These measurements are optimal in achieving the minimal statistical error without adaptive measurements. We then 
determine all  measurements that are Fisher symmetric for the completely mixed state and  all pure states.  
Interestingly, the minimal optimal measurements  are  tied to SICs \cite{Zaun11,ReneBSC04}, which reveals a deep connection between local symmetry and global symmetry. For a qubit, we determine all measurements on a pair that are universally Fisher symmetric for all  states and show that   they are  significantly more efficient than all local  measurements. This prediction was successfully verified in experiments recently \cite{HouTSZ17}.
In the  study, we derive a fundamental constraint on the Fisher information matrix of any collective measurement on a pair, 
which offers a useful tool for characterizing the efficiency of collective measurements.

\emph{Preliminaries.}---Suppose the quantum state $\rho(\theta)$ acting on the $d$-dimensional Hilbert space $\caH$ is characterized by  parameters $\theta_1, \theta_2, \ldots, \theta_g$, where $g=2d-2$ for pure states and $g=d^2-1$ for mixed states. A POVM is a set of positive operators that sum up to the identity. 
Given a POVM $\{\Pi_\xi\}$ on $\caH^{\otimes N}$ with positive integer $N$, the
probability of obtaining the outcome $\xi$ is
$p_\xi(\theta)=\tr[\rho(\theta)^{\otimes N}\Pi_\xi]$. The  \emph{Fisher
information matrix} $I^{(N)}(\theta)$ has matrix elements
\begin{equation}
I^{(N)}_{ab}(\theta)=\sum_{\xi,p_\xi>0}\frac{1}{p_\xi}\frac{\partial p_\xi}{\partial \theta_a}\frac{\partial p_\xi}{\partial \theta_b}.
\end{equation}
The inverse Fisher information matrix
sets a lower bound for the MSE matrix $C^{(N)}(\theta)$ (also known as the covariance matrix) of any unbiased estimator, which is known as the Cram\'er-Rao bound \cite{Rao02book}; see the supplement, which includes \rscite{HayaM08,KahnG09,Naga89A,Haya97inbook,BengZ06book,BarnG00,DankCEL09,GrosAE07,RoyS09, Appl05, Zhu10, HughS16,Szol14}.

The \emph{quantum Fisher information  matrix}  $J(\theta)$ \cite{Hels76book, Hole82book,BrauC94,Haya05book} has matrix elements,
\begin{equation}
J_{ab}(\theta)=\frac{1}{2}\tr\bigl [\rho(L_a L_b+L_bL_a )\bigr ],
\end{equation}
where  the Hermitian operator $L_a $  satisfies the equation $\frac{1}{2}(\rho
L_a+L_a\rho)=\rho_{,a}:=\partial
\rho(\theta)/\partial \theta_a$ and is known as  the symmetric logarithmic derivative associated with $\theta_a$. The matrix $J(\theta)$ is an upper bound for the scaled Fisher information matrix $I^{(N)}(\theta)/N$, so its inverse is a lower bound for the scaled  MSE matrix $NC^{(N)}(\theta)$, which is known as the quantum Cram\'er-Rao bound \cite{Hels76book, Hole82book,Haya05book, BrauC94}. The  bound
generally cannot be saturated except when different $L_a $ can be measured simultaneously.

Another fundamental constraint on the Fisher information matrix $I^{(N)}(\theta)$ is the following inequality derived by Gill and Massar (GM)~\cite{GillM00},
\begin{equation}\label{eq:GMineq}
\tr[ J^{-1}(\theta) I^{(N)}(\theta) ]\leq N(d-1),
\end{equation}
where $\tr [J^{-1}(\theta) I(\theta)]$ is independent of the parametrization. For pure states, the GM inequality applies to arbitrary measurements \cite{GillM00,Mats02}. For mixed states, it applies to arbitrary separable measurements, but may fail for certain collective measurements \cite{GillM00,Zhu12the}. When $N=1$, the GM inequality is saturated iff the  POVM $\{\Pi_\xi\}$ is rank one and  $\tr[\rho(\theta)\Pi_\xi]\neq 0$; see the supplement. In addition, the GM inequality is useful to studying uncertainty relations~\cite{Zhu15IC} and quantum steering~\cite{ZhuHC16}.

A POVM  on $\caH^{\otimes t}$ with positive integer $t$ is \emph{weakly Fisher symmetric} for $\rho(\theta)$ if $I^{(t)}(\theta)$ is proportional to $J(\theta)$ and  \emph{Fisher symmetric}  if  $\tr [J^{-1}(\theta) I^{(t)}(\theta)]$ further attains the maximum over all POVMs on $\caH^{\otimes t}$ \cite{LiFGK16}. Such POVMs are as efficient as possible for estimating all parameters  of interest.
Pure states are specified by $2(d-1)$ parameters, and the inequality in \eref{eq:GMineq} with $N=t$ holds for any POVM on $\caH^{\otimes t}$, 
so  Fisher symmetry means $I^{(t)}(\theta)=\frac{t}{2} J(\theta)$. 
In that  case, each parameter is determined with half  of the maximum resolution for determining this parameter separately. 
When $t=1$, Fisher symmetry for a  mixed state means $I(\theta) =\frac{1}{d+1}J(\theta)$. Such measurements, if they exist, are optimal in minimizing the mean square Bures distance (MSB) and mean infidelity \cite{Zhu12the,HouZXL16}; see the supplement.

When $t=1$,  Fisher-symmetric
measurements have been constructed by Li et al. for any given pure state~\cite{LiFGK16}. In addition, the covariant measurement composed of all pure states weighted by the Haar measure is simultaneously Fisher symmetric for all pure states \cite{Haya98,Zhu12the,Zhu14IOC}. 
However, no measurement with a finite number of outcomes has this property  as shown in \thref{thm:UFSnogo} below, which is proved in the supplement.  
For a mixed state, usually there is no Fisher-symmetric measurement except for the completely mixed  state and qubit states \cite{LiFGK16,Zhu12the,HouZXL16}. 
Even in the case of a qubit, no measurement is Fisher symmetric for all states; even the covariant measurement is far from being Fisher symmetric for a generic state \cite{Zhu14IOC}. 

\begin{theorem}\label{thm:UFSnogo}
No measurement on $\caH$ with a finite number of outcomes is  Fisher symmetric for all pure states. 
\end{theorem}

Before studying collective measurements, we need to introduce several additional concepts. 
A weighted set of pure quantum states $\{|\psi_\xi\>, w_\xi\}$   in  $\caH$ with $w_\xi\geq0$ is  a (weighted) \emph{$t$-design} \cite{Zaun11,ReneBSC04,Scot06} if 
$\sum_{\xi} w_\xi (|\psi_\xi\>\<\psi_\xi|)^{\otimes t }$
is proportional to the projector $P_+
$  onto the  symmetric subspace $\caH_+$ of $\caH^{\otimes t}$.  We are mostly interested in 2-designs and will assume $t=2$ in the following discussion except when stated otherwise. Any 2-design $\{|\psi_\xi\>, w_\xi\}$ has at least $d^2$ elements, and the lower bound is saturated iff all weights $w_\xi$ are equal, and $\{|\psi_\xi\>\}$ forms a SIC \cite{Zaun11,ReneBSC04, Scot06,ApplFZ15G}, 
\begin{equation}
|\<\psi_\xi|\psi_\eta\>|^2=\frac{d\delta_{\xi\eta}+1}{d+1},\quad \forall \xi, \eta. 
\end{equation}

\emph{Fisher-symmetric measurements for pure states.}---Now we are ready to construct measurements  on $\caH^{\otimes 2}$ that are 
Fisher symmetric for all pure states. Since $\rho^{\otimes 2}$ is supported on the symmetric subspace $\caH_+$ whenever $\rho$ is pure, it suffices to construct a POVM on this subspace. Let $\{|\psi_\xi\rangle, w_\xi\}$ be a 2-design with $\sum_{\xi}w_\xi=d(d+1)/2$. Then the operators  $\Pi_\xi=w_\xi (|\psi_\xi\rangle\langle \psi_\xi|)^{\otimes 2}$ form a POVM  on $\caH_+$. Such POVMs have been studied before   and are known to be optimal on average for certain estimation problems \cite{MassP95, HayaHH05}. 
However, little is known about their performance with regard to local tomography. Here we shall show  that these POVMs are optimal for every parameter point simultaneously. Note that the  existence of such  an efficient POVM itself is highly nontrivial. 
\begin{theorem}
 Let $\{|\psi_\xi\rangle, w_\xi\}$ be a 2-design with $\sum_{\xi}w_\xi=d(d+1)/2$ and $\Pi_\xi=w_\xi (|\psi_\xi\rangle\langle \psi_\xi|)^{\otimes 2}$. Then the POVM $\{\Pi_\xi\}$ on  $\caH_+$ is Fisher symmetric for all pure states.
\end{theorem}

\begin{proof}
The probability of obtaining outcome $\xi$  of $\{\Pi_\xi\}$ is $p_\xi=\tr(\rho^{\otimes 2}\Pi_\xi)=w_\xi(\langle\psi_\xi|\rho|\psi_\xi\rangle)^2$, which is factorized. 
These probabilities determine  probabilities associated with the companion POVM  $\bigl\{\frac{2w_\xi}{d+1} (|\psi_\xi\rangle\langle \psi_\xi|)\bigr\}$ on $\caH$, which is  IC. So $\{\Pi_\xi\}$ is also  IC. The Fisher information matrix $I^+$ provided by $\{\Pi_\xi\}$ has matrix elements
\begin{align}\label{eq:FisherSym}
I^+_{ab}&=4\sum_{\xi}w_\xi\langle\psi_\xi|
\rho_{,a}|\psi_\xi\rangle\langle\psi_\xi| \rho_{,b}|\psi_\xi\rangle  \nonumber \\
&=4\tr\bigl[(\rho_{,a}\otimes \rho_{,b})P_+\bigr] = 2\tr (\rho_{,a} \rho_{,b}),
\end{align}
note that $\rho_{,a}$ are traceless and that $\langle\psi_\xi|
\rho_{,a}|\psi_\xi\rangle=0$ whenever $\langle\psi_\xi|
\rho|\psi_\xi\rangle=0$. Interestingly,
the Fisher information matrix is independent of the specific measurement, as long as $\{|\psi_\xi\rangle, w_\xi\}$ is a 2-design. In particular, it is invariant under
the unitary transformation $\Pi_\xi\rightarrow U^{\otimes 2}\Pi_\xi
(U^{\otimes 2})^\dag$ for any unitary $U$ on $\caH$. Therefore, to show that $\{\Pi_\xi\}$ is  Fisher symmetric for all pure states, it suffices to consider any given pure state, say $\rho=|0\rangle\langle 0|$, assuming $|j\rangle$ for $j=0,1,\ldots, d-1$ form an orthonormal basis. 
For pure states, we can choose a suitable parametrization such that $\rho_{,a}$ take on the  form \cite{GillM00}
\begin{equation}
\rho_{,a}=\begin{cases}
|a\rangle\langle 0|+|0\rangle\langle a| & 1\leq a\leq d-1,\\
\rmi(|a'\rangle\langle 0|-|0\rangle\langle a'|) &d\leq a\leq 2(d-1),
\end{cases}
\end{equation}
where $a'=a-d+1$. 
Then $L_a=2\rho_{,a}$ and
\begin{equation}\label{eq:Fisher2designPure}
I^+_{ab}=J_{ab}=4\delta_{ab}.
\end{equation}
So  $\{\Pi_\xi\}$ is  Fisher symmetric for all pure states. 
\end{proof}

If a SIC exists in dimension $d$ \cite{FuchHS17,Zaun11,ReneBSC04}, then we can construct a   Fisher-symmetric measurement for all pure states with only $d^2
$ outcomes. In every prime power dimension, such a measurement can be constructed using a complete set of mutually unbiased bases (MUB) \cite{WootF89,DurtEBZ10,Zhu15M}, which forms a 2-design with $d^2+d$ elements. In general, let $d'$ be the smallest prime power that is not smaller than $d$ (which satisfies $d'\leq 2d-2$); then a 2-design in dimension $d$ can be constructed by projecting  a complete set of mutually unbiased bases in dimension $d'$ to a subspace of dimension $d$. So  we can always  construct  a   Fisher-symmetric measurement for all pure states in  dimension~$d$ with no more than 
 $4d^2$ outcomes. It is worth pointing out that tensor products of POVMs constructed above are also Fisher symmetric for all pure states. So are POVMs on $\caH^{\otimes t}$
 constructed from $t$-designs with $t\geq3$. However, such POVMs offer little advantage over those constructed from 2-designs, but are much more difficult to implement.

\emph{Fisher-symmetric measurements for mixed states.}---Here we need to generalize the concepts of 2-designs and SICs. A set of positive operators $\{\Pi_\xi\}$ is called a \emph{generalized 2-design} if 
\begin{align}\label{eq:g2design}
\sum_{\xi}\frac{ \Pi_\xi \otimes \Pi_\xi}{\tr(\Pi_\xi)}=&\frac{\sum_{\xi} w_\xi}{d}\Bigl(\frac{1+\wp}{d+1}P_++\frac{1-\wp}{d-1}P_-\Bigr), 
\end{align}
where  $w_\xi=\tr(\Pi_\xi)$, 
$\wp=\sum_\xi w_\xi \wp_\xi/(\sum_{\xi} w_\xi)$,  $\wp_\xi= \tr(\Pi_\xi^2)/(\tr\Pi_\xi)^2$, and $P_-$ is the projector onto the antisymmetric subspace $\caH_-$ of $\caH^{\otimes 2}$; cf.~\rcite{GrayA16}. Here $\wp_\xi$ may be interpreted as the purity of $\Pi_\xi$, and $\wp$ as the  purity of the set $\{\Pi_\xi\}$.  A set  of $d^2$ positive operators  $\{\Pi_\xi\}$ is  \emph{a generalized SIC} if $\sum_{\xi}\Pi_\xi$ is proportional to the identity and 
$\tr(\Pi_\xi\Pi_\eta)=\alpha \delta_{\xi\eta}+\beta$
for some positive constants $\alpha,\beta$ \cite{Appl07,GourK14,Zhu14T}. Any generalized SIC is a generalized 2-design; see the supplement for a partial converse.

No  measurement on $\caH$ is Fisher symmetric for a mixed state $\rho$
except when $\rho$ is the completely mixed state or a qubit state \cite{LiFGK16}. In preparation for later applications, \pref{pro:WFSmix} and \crref{cor:FSmixMinOne} below  clarify the structure of (weakly) Fisher-symmetric measurements at the completely mixed state. The proofs are relegated to  the supplement, which also explains the connection with tight IC measurements introduced by Scott~\cite{Scot06}. 
\begin{proposition}\label{pro:WFSmix}
	A  POVM $\{\Pi_\xi\}$ on $\caH$ is (weakly) Fisher symmetric at the completely mixed state iff $\{\Pi_\xi\}$ is a (generalized) 2-design.
\end{proposition}

\begin{corollary}\label{cor:FSmixMinOne}
	Any   POVM $\{\Pi_\xi\}$ on $\caH$ that is Fisher symmetric at the completely mixed state has at least $d^2$ elements; the lower bound is saturated iff $\{\Pi_\xi\}$ is a SIC.  
\end{corollary}

It is much more difficult to study  Fisher-symmetric measurements   for mixed states when  $t\geq2$, because the GM inequality does not apply to collective measurements, and an extension of the GM inequality  has been a long-standing open problem.
Nevertheless, we have a simple solution in the case $t=2$. 
For simplicity,  $\rho(\theta)$ has full rank in the rest of the paper. 
\begin{theorem}\label{thm:GMT2}
The Fisher information matrix $I^{(2)}(\theta)$ at $\rho(\theta)$ of any POVM $\{\Pi_\xi\} $  on $\caH^{\otimes 2}$ satisfies
	\begin{equation}\label{eq:GMTmax}
	\tr [J^{-1}(\theta) I^{(2)}(\theta)] \leq 3d-3. 
	\end{equation}
	The inequality is saturated iff each $\Pi_\xi $ is proportional to either the tensor power of a pure state or a Slater-determinant state.  
\end{theorem}
A variant of \thref{thm:GMT2} was proved in the thesis of the first author \cite{Zhu12the};
see the supplement for a simplified proof. 	
Here a Slater-determinant state has the form $U^{\otimes 2}|\Psi_-\>\<\Psi_-|(U^{\otimes 2})^\dag$, where $|\Psi_-\>=(|01\>-|10\>)/\sqrt{2}$, and $U$ is a unitary. Both Slater-determinant states and tensor powers of pure states are generalized coherent states~\cite{ZhanFG90}, which are least  entangled and most classical for the given symmetry.  Measurements (POVMs) composed of these states  are referred to as \emph{coherent measurements} (POVMs) henceforth. The inequality in \eref{eq:GMTmax} is saturated iff the POVM is coherent. Some POVMs known in the literature \cite{VidaLPT99,GillM00}
are  coherent by our definition, although they were introduced for different purposes.

\Thref{thm:GMT2} implies  $\tr [J^{-1}(\theta) I^{(N)}(\theta)] \leq 3N(d-1)/2$ if each time we can measure at most two copies of $\rho(\theta)$ together, that is, if the POVM on $\caH^{\otimes N}$ is a tensor product of POVMs on $\caH$ or $\caH^{\otimes 2}$.
Compared with the GM inequality in  \eref{eq:GMineq}, here the upper bound is $50\%$ larger, which reflects the advantage of collective measurements over separable measurements. The following corollary is a tomographic implication of \thref{thm:GMT2} and the analog of the GM bound in Eq.~(30) of \rcite{GillM00}  (note that a factor of $1/(d-1)$ is missing there); see the supplement on the GM bound.  
\begin{corollary}\label{cor:WMSEent2}
In quantum state tomography with any collective measurement on $\caH^{\otimes 2}$, the scaled 
weighted mean square error (WMSE) $N\tr(WC^{(N)})$  of any unbiased estimator is bounded from below by
\begin{equation}\label{eq:WMSE2copy}
\mse_{W}=\frac{2\bigl(\tr\sqrt{J^{-1/2}WJ^{-1/2}}\,\bigr)^2}{3(d-1)}.
\end{equation}
The  bound can be saturated iff there exists a measurement on $\caH^{\otimes 2}$ that yields the Fisher information matrix
\begin{equation}\label{eq:FisherOpt2copy}
I^{(2)}_W=3(d-1)J^{1/2}\frac{\sqrt{J^{-1/2}WJ^{-1/2}}}{\tr\sqrt{J^{-1/2}WJ^{-1/2}}}J^{1/2}. 
\end{equation}
\end{corollary}
Here $W$ is a positive semidefinite matrix that may depend on the parameter point. 
For  MSB,  $W=J/4$~\cite{BrauC94}, so the bound in \eref{eq:WMSE2copy} is saturated iff the measurement yields $I^{(2)}=3J/(d+1)$ and   is thus Fisher symmetric.

Any coherent POVM  $\{\Pi_\xi\}$ is  the union of two POVMs $\{\Pi_\zeta^+\}$ and $\{\Pi_\eta^-\}$ on the symmetric and antisymmetric subspaces, respectively. The POVM $\{\Pi_\zeta^+\}$ is tied to a 2-design and is thus Fisher symmetric for all pure states. Its contribution to the Fisher information matrix  is independent of the specific POVM $\{\Pi_\zeta^+\}$ by \eref{eq:FisherSym}. For a qubit,   all $\Pi_\eta^-$ are proportional to the singlet  $|\Psi_-\>\<\Psi_-|$, so all coherent   POVMs yield the same Fisher information matrix.  Since any 2-design has at least $d^2$ elements, and minimal 2-designs are in one-to-one correspondence with SICs \cite{ReneBSC04, Scot06,ApplFZ15G,Zhu14T}, we deduce the following corollary.
\begin{corollary}\label{cor:CohPOVMmin}
Any coherent POVM  on $\caH^{\otimes 2}$  has at least $\frac{1}{2}(3d^2-d)$  elements. The  bound is saturated iff $\frac{1}{2}d(d-1)$ elements are Slater-determinant states and form a projective measurement on $\caH_-$, and the other $d^2$ elements have the form $\frac{d+1}{2d}(|\psi_\xi\>\<\psi_\xi|)^{\otimes 2}$, where $\{|\psi_\xi\>\}$ is a SIC. 
\end{corollary}

\begin{figure}
	\includegraphics[width=8.3cm]{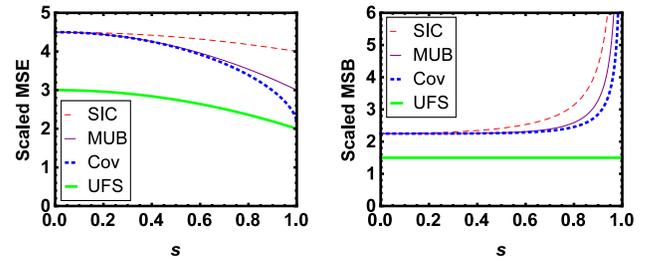}
	\caption{\label{fig:UFS}(color online)
		Scaled MSE (left plot) and scaled MSB (right plot) achieved by universally Fisher-symmetric measurements (UFS), including the collective SIC,  in qubit state tomography. Here $s$ is the length of the Bloch vector. The performances of  SIC, MUB, and covariant measurements (averaged over states with the same purity) are shown for comparison (reproduced from \rcite{Zhu14IOC}). 	}
\end{figure}

In the case of a qubit, a minimal coherent POVM  has five elements, which take on the  form \cite{VidaLPT99,Zhu12the}
\begin{equation}\label{eq:CollectiveSIC}
\Pi_\xi =\frac{3}{4}(|\psi_\xi\>\<\psi_\xi|)^{\otimes 2}, \quad \Pi_5=|\Psi_-\>\<\Psi_-|, 
\end{equation}
where $|\psi_\xi\>$ for $\xi=1,2,3,4$ form a SIC, and $|\Psi_-\>$ is the singlet. This POVM is referred to as the \emph{collective SIC} henceforth. Remarkably, it is universally Fisher symmetric, that is, Fisher symmetric   for all  states. To see this, parametrize the qubit state $\rho$  by the Bloch vector $\vec{s}=(s_1,s_2,s_3)$ as $\rho=\frac{1}{2}(1+\vec{s}\cdot \vec{\sigma})$. Then 
\begin{equation}
I^{(2)}_{ab}=J_{ab}=\delta_{ab}+\frac{s_as_b}{1-s^2},\quad a,b=1,2,3.
\end{equation} In conjunction with \thref{thm:GMT2}, we deduce the following.
\begin{theorem}
When $d=2$, a POVM $\{\Pi_\xi\}$ on $\caH^{\otimes2}$    is universally Fisher symmetric iff it is coherent. 
\end{theorem}

The scaled MSE (with respect to the Hilbert-Schmidt distance) and scaled MSB achieved by the collective SIC are respectively given by
\begin{equation}
\mse(\rho)=3-s^2,\quad \msb(\rho)=\frac{3}{2},
\end{equation}
as illustrated in \fref{fig:UFS}. 
The tomographic efficiency is much higher than all POVMs on individual copies. 
In particular, the scaled MSB achieved by any fixed measurement on individual copies diverges in the pure-state limit \cite{Zhu12the, Zhu14IOC,HouZXL16}; accordingly,  the mean infidelity scales as $O(1/\sqrt{N})$  \cite{MahlRDF13}. 
By contrast, the scaled MSB achieved by the collective SIC saturates the bound in \eref{eq:WMSE2copy} with $W=J/4$, so that the mean infidelity scales as  $O(1/N)$. Recently,  the collective SIC was successfully realized in  experiments, which 
achieved the highest tomographic efficiency in qubit state tomography to date~\cite{HouTSZ17}.

In general, to be Fisher symmetric, the Fisher information matrix should equal
 $I^{(2)}=3J/(d+1)$ according to \thref{thm:GMT2}. In the limit to pure states, this requirement is not compatible with \eref{eq:Fisher2designPure} when $d\geq 3$. Therefore, we believe that Fisher-symmetric measurements in general cannot exist for mixed states when $d\geq3$. 
Nevertheless, it is still desirable to construct POVMs that are Fisher symmetric 
for the completely mixed state and for all pure states simultaneously. Such POVMs are
called \emph{tight}; they are optimal in the tomography of pure states and highly mixed states. 
According  to \thref{thm:GMT2}, all tight POVMs on $\caH^{\otimes 2}$ are  coherent POVMs, which are automatically Fisher symmetric 
for all pure states. \Thsref{thm:UFS}
and \ref{thm:UFSmin} below clarify the structure of such POVMs; the proofs are relegated to the supplement.

\begin{theorem}\label{thm:UFS}
A   POVM   $\{\Pi_\xi\}$ on $\caH^{\otimes 2}$ is tight coherent (Fisher symmetric for the completely mixed state and all pure states)
iff $\{Q_\xi\}$ is a generalized 2-design of purity $\frac{3d+1}{4d}$, where   $Q_\xi=\tr_1(\Pi_\xi)+\tr_2(\Pi_\xi)$. 
\end{theorem}

\Thref{thm:UFS} offers a recipe for creating tight coherent POVMs. Let $\{A_\zeta\}$ be a 2-design with  $\sum_\zeta A_\zeta=(d+1)/2$ and $\{B_\eta\}$  a generalized 2-design with $\sum_\eta B_\eta=2(d-1)$ and with $B_\eta$  proportional to rank-2 projectors. Then the union of $\{\Pi_\zeta^+ \}$ and $\{\Pi^-_\eta\}$ is tight coherent, where 
\begin{equation}\label{eq:UFSconstruct}
\Pi_\zeta^+ =\frac{A_\zeta\otimes A_\zeta}{\tr(A_\zeta)}, \quad \Pi^-_\eta=\frac{P_-(B_\eta\otimes B_\eta)P_-}{\tr(B_\eta)}.
\end{equation}

\begin{theorem}\label{thm:UFSmin}
Any tight coherent  POVM $\{\Pi_\xi\}$ on $\caH^{\otimes2}$ has at least $2d^2$ elements when $d\geq 3$. The lower bound is saturated iff $\{\Pi_\xi\}$ is a union of two POVMs $\{\Pi_\zeta^+ \}$ and $\{\Pi_\eta^- \}$ on $\caH_+$ and $\caH_-$, respectively, $\{Q_\zeta^+  \}$ forms a SIC, and $\{Q_\eta^- \}$ forms a generalized SIC of purity $\frac{1}{2}$.
\end{theorem}
\Thref{thm:UFSmin} offers a general recipe for constructing minimal tight coherent  POVMs. When $d=3$, interestingly,    such POVMs are in one-to-one correspondence with pairs of SICs (see the supplement).

\emph{Summary.}---We  introduced a general method for constructing two-copy collective measurements   that are universally  Fisher symmetric for all pure states. These measurements are  optimal in achieving the minimal statistical error without  adaptive measurements. 
We also determined all collective measurements on a pair that are Fisher symmetric for the completely mixed state
and  for all pure states.  For a qubit,  they are  Fisher symmetric for all  states and are substantially more efficient than all local measurements. In the study, we derived a fundamental constraint on the Fisher information matrix of any collective measurement on a pair, which provides a useful tool for characterizing the power of collective measurements. 
Our work is of interest not only to studying  quantum measurements and estimation theory, but also to improving  efficiency and precision in practical quantum state tomography and multiparameter quantum metrology.

\bigskip

\acknowledgments
HZ is grateful to  Carlton M. Caves for discussions and to Blake Stacey  for comments.
HZ acknowledges financial support from the Excellence Initiative of the German Federal and State Governments Zukunftskonzept
(ZUK~81) and the Deutsche Forschungsgemeinschaft (DFG).
MH was supported in part by  Japan Society for the Promotion of 
Science (JSPS) Grant-in-Aid for Scientific Research (A) No. 17H01280, (B) No. 16KT0017, and Kayamori Foundation of Informational Science Advancement.


%

\nocite{apsrev41Control}
\bibliographystyle{apsrev4-1}
\bibliography{all_references}

\clearpage
\newpage
\addtolength{\textwidth}{-0.9in}
\addtolength{\oddsidemargin}{0.45in}
\addtolength{\evensidemargin}{0.45in}	
	

	\setcounter{equation}{0}
	\setcounter{figure}{0}
	\setcounter{table}{0}
	\setcounter{theorem}{0}
	\setcounter{lemma}{0}
	\setcounter{remark}{0}
	\setcounter{proposition}{0}
	\setcounter{corollary}{0}
	\setcounter{section}{0}

	\makeatletter
	\renewcommand{\theequation}{S\arabic{equation}}
	\renewcommand{\thefigure}{S\arabic{figure}}
	\renewcommand{\thetable}{S\arabic{table}}
	\renewcommand{\thetheorem}{S\arabic{theorem}}
	\renewcommand{\thelemma}{S\arabic{lemma}}
	\renewcommand{\theremark}{S\arabic{remark}}
	\renewcommand{\theproposition}{S\arabic{proposition}}
	\renewcommand{\thecorollary}{S\arabic{corollary}}
	
	
 \onecolumngrid
 \begin{center}
 	\textbf{\large Universally  Fisher-Symmetric Informationally Complete Measurements: Supplement}
 \end{center}
	
	In this supplement, we prove \thsref{thm:UFSnogo}, \ref{thm:GMT2}, \ref{thm:UFS}, \ref{thm:UFSmin},  \pref{pro:WFSmix}, and \crref{cor:FSmixMinOne} presented in the main text. We also provide more details on quantum state tomography with collective measurements, the Gill-Massar inequality, Gill-Massar bound \cite{GillM00,Zhu12the,HouZXL16}, generalized 2-designs, generalized   symmetric informationally complete measurements (SICs for short) \cite{Zaun11,ReneBSC04, Appl07,GourK14,Zhu14T}, and tight coherent measurements in dimension 3. 
	
	\section{\label{sec:QSTcoll}Quantum state tomography with collective measurements} 
	Quantum state tomography is a procedure for inferring the state of a quantum system from statistics of quantum measurements \cite{Hels76book, Hole82book, Haya05book, PariR04Book,LvovR09}. To achieve sufficient precision, usually many identically prepared quantum systems need to be measured. The simplest measurement strategy is to repeat a given measurement $N$ times when $N$ identically prepared systems are available for tomography, as illustrated in the left plot of \fref{fig:QST}. In this case, the efficiency of the quantum measurement is mainly determined by the Fisher information matrix since its inverse sets a lower bound for the mean-square-error (MSE) matrix of any unbiased estimator. In addition, the lower bound can be saturated asymptotically by the  maximum-likelihood estimator \cite{Rao02book,PariR04Book}.  The Fisher-information matrices for independent measurements  are additive, which means the MSE achievable by repeated measurements is inversely proportional to the sample size $N$ when the measurement is informationally complete (IC) with regard to the parameters of interest.

	Repetition of a fixed measurement on individual quantum systems is not so efficient when the mean infidelity or mean square Bures distance (MSB) is the figure of merit \cite{MahlRDF13,Zhu12the,Zhu14IOC,HouZXL16}.
	Adaptive measurements on individual quantum systems can improve the tomographic efficiency to some degree. 
	To achieve the optimal performance, however, usually one needs  to perform collective measurements on all $N$ quantum systems together \cite{MassP95,VidaLPT99,HayaHH05} (see the middle plot of \fref{fig:QST}), which is usually not realistic in practice. 
	In the large-$N$ limit, 
	the optimal performance is determined by 
	the quantum Fisher information matrix \cite{HayaM08,KahnG09,Zhu12the}. To determine the optimal performance for a given sample size $N$, most researchers have adopted Bayesian approaches and employed the mean fidelity as the figure of merit \cite{MassP95,VidaLPT99,HayaHH05,Haya98}.

	\begin{figure}[htbp]
		\begin{center}
			\includegraphics[ width=0.99\linewidth]{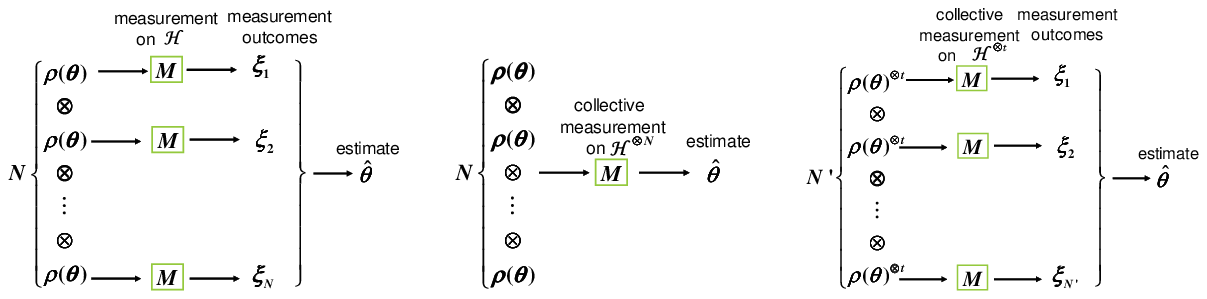}
			\caption{Comparison of three  measurement schemes for quantum state tomography. Left plot: repeated individual measurements; middle plot: single collective measurement on all quantum systems; right plot: repeated collective measurements on a limited number of quantum systems.   }
			\label{fig:QST}
		\end{center}
	\end{figure}

	In this paper we are interested in the scenario in which we can perform limited collective measurements~\cite{Zhu12the}. 
	For example, 
	suppose $N=t N'$; 
	we measure $t$ identically prepared quantum systems together,
	and repeat this procedure $N'$ times, as illustrated in the right plot of \fref{fig:QST}.
	Such scenarios are more accessible to experiments and are sufficient to demonstrate the key distinction between collective measurements and individual measurements. 
	
	Suppose  the  density matrix $\rho(\theta)$ is characterized by a set of parameters denoted collectively by $\theta$, 
	and we are interested in estimating these parameters.
	If we perform a collective measurement  described by the POVM $M=\{\Pi_\xi\}$
	on  ${\cal H}^{\otimes t}$, then the probability of obtaining outcome $\xi$ is given by
	$p_\xi(\theta):=\tr[\rho(\theta)^{\otimes t}\Pi_\xi]$. After the measurement is repeated  $N'=N/t$ times (see the right plot of \fref{fig:QST}),  the total Fisher information matrix is given by $I^{(N)}(\theta)=N'I^{(t)}(\theta)$, with
	\begin{equation}\label{eq:FisherSupp}
	I^{(t)}_{ab}(\theta)=\sum_{\xi,p_\xi>0}\frac{1}{p_\xi}\frac{\partial p_\xi}{\partial \theta_a}\frac{\partial p_\xi}{\partial \theta_b}.
	\end{equation}
	The MSE matrix $C^{(N)}(\theta)$  of any unbiased estimator is bounded from below by the inverse Fisher information matrix $[I^{(N)}(\theta)]^{-1}=[N'I^{(t)}(\theta)]^{-1}$ \cite{Rao02book}. When $N'$ is sufficiently large, the maximum-likelihood estimator can approximately saturate this lower bound, and the corresponding  scaled MSE matrix reads $NC^{(N)}(\theta)\approx t[I^{(t)}(\theta)]^{-1}$. To achieve high tomographic efficiency, the main task is to construct collective measurements that yield the most Fisher information.

	Here, we should remark why we do not count the event $\xi'$ with $p_{\xi'}(\tilde{\theta})=0$
	in the definition of the Fisher information matrix $I^{(t)}(\theta)$  at $\tilde{\theta} $, as manifested in \eref{eq:FisherSupp}, assuming that $\rho(\theta)$ is pure. Note that the parametrization $ \rho(\theta)$ is differentiable, so
	$p_{\xi'} (\theta) $ is also differentiable. In addition, $\rho(\theta)$ is a pure state so it can be written as $\rho(\theta)=|\psi(\theta)\>\<\psi(\theta)|$, which implies that $\frac{\partial p_{\xi'} (\theta)}{\partial \theta_a}\big|_{\theta=\tilde{\theta}} =0$  whenever $p_{\xi'}(\tilde{\theta})=0$. 
	
	We shall illustrate our argument by considering the estimation problem  of one parameter denoted by $\theta$. The significance of the Fisher information $I(\theta)$ is tied to the mean square error (MSE) of a locally unbiased estimator. 
	More precisely, its inverse sets a lower bound for the MSE of any  locally unbiased estimator $\hat{\theta}$ of $\theta$, which is known as the Cram\'er-Rao bound, and
	the bound can be saturated asymptotically by the maximum-likelihood estimator \cite{Rao02book, PariR04Book}. This fact underpins the definition of the Fisher information.

	Recall that an estimator $\hat{\theta}$ of $\theta$ is a mapping from the data to the parameter $\xi\to \hat{\theta}(\xi)$. Its MSE reads
	\begin{equation}
	C(\hat{\theta},\theta):=\sum_{\xi}p_{\xi}(\theta)[\hat{\theta}(\xi)-\theta]^2.
	\end{equation}
	The estimator $\hat{\theta}$ is locally unbiased at $\tilde{\theta}$ \cite{Hole82book} if
	\begin{equation}
	\sum_{\xi}\hat{\theta}(\xi) p_{\xi}(\tilde{\theta})=
	\tilde{\theta}, \quad 
	\sum_{\xi}\hat{\theta}(\xi) \frac{\partial p_{\xi} (\theta)}{\partial \theta}\Big|_{\theta=\tilde{\theta}}=1.
	\end{equation}
	When $ p_{\xi'} (\tilde{\theta})=0$ and $\frac{\partial p_{\xi'} (\theta)}{\partial \theta}\big|_{\theta=\tilde{\theta}}=0$, 
	the terms corresponding to $\xi'$ 
	can be eliminated from the above summation, so we get 
	\begin{equation}
	\sum_{\xi,p_{\xi}(\tilde{\theta})>0}\bigl[\hat{\theta}(\xi)-\tilde{\theta}\bigr]p_{\xi}(\tilde{\theta})^{1/2} p_{\xi}(\tilde{\theta})^{-1/2}\frac{\partial p_{\xi} (\theta)}{\partial \theta}\Big|_{\theta=\tilde{\theta}}=\sum_{\xi}\bigl[\hat{\theta}(\xi)-\tilde{\theta}\bigr] \frac{\partial p_{\xi} (\theta)}{\partial \theta}\Big|_{\theta=\tilde{\theta}}=\sum_{\xi}\hat{\theta}(\xi) \frac{\partial p_{\xi} (\theta)}{\partial \theta}\Big|_{\theta=\tilde{\theta}}=1.
	\end{equation}
	Applying the Cauchy-Schwarz inequality yields
	\begin{equation}
	\Biggl\{\sum_{\xi,p_{\xi}(\tilde{\theta})>0}p_{\xi}(\tilde{\theta}) \bigl[\hat{\theta}(\xi)-\tilde{\theta}\bigr]^2\Biggr\}\Bigg\{\sum_{\xi,p_{\xi}(\tilde{\theta})>0} \frac{1}{p_{\xi}(\tilde{\theta})}\biggl(\frac{\partial p_{\xi} (\theta)}{\partial \theta}\Big|_{\theta=\tilde{\theta}}\biggr)^2\Biggr\}\geq 1,
	\end{equation}
	which implies that
	\begin{equation}
	C(\hat{\theta},\tilde{\theta})=\sum_{\xi}p_{\xi}(\tilde{\theta}) \bigl[\hat{\theta}(\xi)-\tilde{\theta}\bigr]^2=\sum_{\xi,p_{\xi}(\tilde{\theta})>0}p_{\xi}(\tilde{\theta}) \bigl[\hat{\theta}(\xi)-\tilde{\theta}\bigr]^2\geq \Bigg\{\sum_{\xi,p_{\xi}(\tilde{\theta})>0} \frac{1}{p_{\xi}(\tilde{\theta})}\biggl(\frac{\partial p_{\xi} (\theta)}{\partial \theta}\Big|_{\theta=\tilde{\theta}}\biggr)^2\Biggr\}^{-1}.
	\end{equation}
	Further, the equality holds when 
	the locally unbiased estimator at $\tilde{\theta}$ is given as \cite{Naga89A,Haya05book}
	\begin{equation}
	\hat{\theta}_{\tilde{\theta}}(\xi):=\frac{1}{I(\tilde{\theta}) p_{\xi}(\tilde{\theta})}
	\frac{\partial p_{\xi} (\theta)}{\partial \theta}\Big|_{\theta=\tilde{\theta}}
	+\tilde{\theta}.\label{eq:La}
	\end{equation}
	One might think that the estimator \eqref{eq:La} is meaningless because 
	it depends on the true parameter $\tilde{\theta}$.
	However, this estimator is asymptotically close to the maximum-likelihood estimator in the following sense.
	The maximum-likelihood estimator $\hat{\theta}_\ml$ is asymptotically
	locally unbiased  at all points
	and attains the Cram\'{e}r-Rao bound.
	When the true parameter is $\tilde{\theta}$ and we observe $N$ outcomes, the variable
	$\sqrt{N} (\hat{\theta}_\ml-\tilde{\theta}) $
	asymptotically approaches  $\sqrt{N}$ times of
	the sample mean of $\hat{\theta}_{\tilde{\theta}}-\tilde{\theta}$,
	which converges to the Gaussian distribution with variance $1/I(\tilde{\theta})$ according to the central limit theorem~\cite{Rao02book}.
	In this way, the asymptotic optimality of 
	the maximum-likelihood estimator can be shown.
	The above discussion on the  MSE of a locally unbiased estimator  explains why the event $\xi'$ with $p_{\xi'}(\tilde{\theta})=0$ and $\frac{\partial p_{\xi'} (\theta)}{\partial \theta}\big|_{\theta=\tilde{\theta}}=0$ does not contribute to the Fisher information  $I(\theta)$ at $\tilde{\theta}$.

	\section{Gill-Massar inequality and Gill-Massar bound}
	In this section we provide more details on the Gill-Massar (GM) inequality and GM bound for the scaled weighted mean square error (WMSE) \cite{GillM00,Zhu12the,HouZXL16}. In particular we provide a self-contained proof of the GM  inequality in the case of one-copy measurement and clarify the equality condition, which is useful to proving \thref{thm:UFSnogo}. In addition, our discussion on the GM bound is  instructive to deriving \crref{cor:WMSEent2} from \thref{thm:GMT2}

	\subsection{Gill-Massar inequality}
	Recall that $\rho(\theta)$ is a quantum state on the Hilbert space $\caH$ of  dimension $d$, $J(\theta)$ is the quantum Fisher information matrix, and $I^{(N)}(\theta)$ is the  Fisher information matrix of a (collective) measurement on $\caH^{\otimes N}$. The GM inequality \cite{GillM00} states that 
	\begin{equation}\label{eq:GMineqSupp}
	\tr [J^{-1}(\theta) I^{(N)}(\theta)] \leq N(d-1)
	\end{equation}
	whenever the measurement is separable. 
	Note  that $\tr [ J^{-1}(\theta) I^{(N)}(\theta) ]$ is invariant under reparametrization. For pure states, which are characterized by $2d-2$ parameters, the GM inequality  applies to arbitrary measurements, not necessarily separable. Here we provide a self-contained proof of the GM inequality in the case $N=1$, assuming that $\rho(\theta)$ is either pure or of full rank. The main purpose of presenting   this proof is to 
	clarify the equality condition of the GM inequality, which was mentioned in \rcite{LiFGK16} without proof.

	\begin{proposition}\label{pro:GMTPure}
		Suppose $\rho(\theta)$ is a pure state parametrized by $\theta_1, \theta_2, \ldots \theta_{2d-2}$.  Then 
		the Fisher information matrix $I(\theta)$ at $\theta$ of any POVM $\{\Pi_\xi\} $  on $\caH$ satisfies
		\begin{equation}
		\tr [ J^{-1}(\theta) I(\theta) ] \leq d-1. 
		\end{equation}
		The inequality is saturated iff all $\Pi_\xi$ are rank one and $\tr[\rho(\theta)\Pi_\xi]> 0$ for all $\xi$. 
	\end{proposition}
	This proposition is crucial to proving \thref{thm:UFSnogo} in the main text, which states that no measurement with a finite number of outcomes on $\caH$ is Fisher symmetric for all pure states.
	\begin{proof}
		The Fisher information matrix $I(\theta)$ has matrix elements
		\begin{align}
		I_{ab}=\sum_{\xi,p_\xi>0}\frac{1}{p_\xi}\frac{\partial p_\xi}{\partial \theta_a}\frac{\partial p_\xi}{\partial \theta_b}=\sum_{\xi,p_\xi> 0}\frac{\tr(\rho_{,a}\Pi_\xi)\tr(\rho_{,b}\Pi_\xi)}{\tr(\rho\Pi_\xi)}=\sum_{\xi,p_\xi> 0}\frac{\tr\bigl[(\rho_{,a}\otimes \rho_{,b})\Pi_\xi^{\otimes 2}\bigr]}{\tr(\rho\Pi_\xi)},
		\end{align}
		where $p_\xi=p_\xi(\theta)=\tr[\rho(\theta)\Pi_\xi]$ and $\rho_{,a}=\rho_{,a}(\theta)=\partial\rho(\theta)/\partial \theta_a$.
		According to \lref{lem:QFIsumPure} below,
		\begin{align}
		\tr(J^{-1}I)&=\sum_{a,b=1}^{2d-2}(J^{-1})_{ab}I_{ab}
		=\sum_{\xi,p_\xi> 0}\frac{\tr\bigl[V(1\otimes \rho+\rho\otimes 1)\Pi_\xi^{\otimes 2}\bigr] -2[\tr(\rho\Pi_\xi)]^2}{2\tr(\rho\Pi_\xi)}\nonumber\\
		&=\sum_{\xi,p_\xi> 0}\biggl[\frac{\tr\bigl(\rho\Pi_\xi^2\bigr)}{\tr(\rho\Pi_\xi)}-\tr(\rho\Pi_\xi)\biggr]
		\leq \sum_{\xi,p_\xi> 0}\tr(\Pi_\xi)-\sum_{\xi}\tr(\rho\Pi_\xi)
		\leq \sum_{\xi}\tr(\Pi_\xi)-1= d-1.
		\end{align}
		Here the first inequality is saturated iff all $\Pi_\xi$ with $p_\xi>0$ have rank one, the second inequality is saturated iff $p_\xi>0$ for all $\xi$. This observation completes the proof of \pref{pro:GMTPure}. 
	\end{proof}

	\begin{proposition}\label{pro:GMTmix}
		Suppose $\rho(\theta)$ is a state of full rank   parametrized by $\theta_1, \theta_2, \ldots \theta_{d^2-1}$.  Then 
		the Fisher information matrix $I(\theta)$ at $\theta$ of any POVM $\{\Pi_\xi\} $  on $\caH$ satisfies
		\begin{equation}
		\tr [ J^{-1}(\theta) I(\theta) ] \leq d-1. 
		\end{equation}
		The inequality is saturated iff all $\Pi_\xi $  are rank one. 
	\end{proposition}
	\begin{proof}
		The proof is similar to that of \pref{pro:GMTPure} except that  \lref{lem:QFIsumPure} employed there should be replaced by  \lref{lem:QFIsumMix}. In addition, the requirement  $p_\xi>0$ is satisfied automatically because $\rho(\theta)$ has full rank. 
	\end{proof}

	\begin{lemma}\label{lem:QFIsumPure}
		Suppose $\rho(\theta)$ is a pure state  parametrized by $\theta_1, \theta_2, \ldots \theta_{2d-2}$.  Then 
		\begin{equation}\label{eq:QFIsumPure}
		\sum_{a,b=1}^{2d-2}(J^{-1})_{ab}(\rho_{,a}\otimes
		\rho_{,b})=\frac{1}{2}V(\rho\otimes 1+1\otimes
		\rho)-\rho^{\otimes2},
		\end{equation}
		where $V=\sum_{j,k} |jk\>\<kj|$ is the swap operator. 
	\end{lemma}
	\begin{proof}
		Note that the left hand side of \eref{eq:QFIsumPure}  is invariant under changes of parametrization. Choose an orthonormal  basis $\{|j\rangle\}_{j=0}^{d-1}$ such that $\rho=|0\rangle\langle 0|$ at the parameter point of interest. Choose a suitable parametrization such that $\rho_{,a}$ take on the  form \cite{GillM00}
		\begin{equation}
		\rho_{,a}=\begin{cases}
		\rho_{,j+} & 1\leq a\leq d-1,\\
		\rho_{,k-}&d\leq a\leq 2(d-1),
		\end{cases}
		\end{equation}
		where $j=a$, $k=a-d+1$, and 
		\begin{equation}
		\rho_{,j+}=
		|j\rangle\langle 0|+|0\rangle\langle j|,\quad 
		\rho_{,j-}=\rmi (|j\rangle\langle 0|-|0\rangle\langle j|), \quad 1\leq j\leq d-1. 
		\end{equation}
		Then the symmetric logarithmic derivatives can be chosen to be $L_a=2\rho_{,a}$, so that 
		\begin{equation}
		J_{ab}=4\delta_{ab}. 
		\end{equation}
		In addition,
		\begin{align}
		\sum_{a=1}^{2d-2}\rho_{,a}\otimes \rho_{,a}&=\sum_{j=1}^{d-1}\left(\rho_{,j+}\otimes \rho_{,j+}+\rho_{,j-}\otimes \rho_{,j-}\right)=2\sum_{j=1}^{d-1}(|0j\>\<j0|+|j0\>\<0j|)\nonumber\\
		&=2\sum_{j=0}^{d-1}(|0j\>\<j0|+|j0\>\<0j|)-4 (|0\>\<0|)^{\otimes 2}=2V(1\otimes\rho+\rho\otimes 1)-4\rho^{\otimes 2}.
		\end{align}
		Therefore, 
		\begin{equation}
		\sum_{a,b=1}^{2d-2}(J^{-1})_{ab}(\rho_{,a}\otimes
		\rho_{,b})=\frac{1}{4}\sum_{a=1}^{2d-2}\rho_{,a}\otimes \rho_{,a}=\frac{1}{2}V(\rho\otimes 1+1\otimes
		\rho)-\rho^{\otimes 2}. 
		\end{equation}
	\end{proof}
	
	\begin{lemma}\label{lem:QFIsumMix}
		Suppose $\rho(\theta)$ is a state of full rank parametrized by $\theta_1, \theta_2, \ldots \theta_{d^2-1}$.  Then 
		\begin{equation}\label{eq:QFIsumMix}
		\sum_{a,b=1}^{d^2-1}(J^{-1})_{ab}(\rho_{,a}\otimes
		\rho_{,b})=\frac{1}{2}V(\rho\otimes 1+1\otimes
		\rho)-\rho^{\otimes2}.
		\end{equation}
	\end{lemma}
	This lemma was reproduced from Lemma~6.2 in \rcite{Zhu12the}. 
	Note that the left hand side of \eref{eq:QFIsumMix}  is invariant under changes of parametrization. In addition, the right hand side of \eref{eq:QFIsumMix} has the same form as that of   \eref{eq:QFIsumPure}.
	\begin{proof}
		We first diagonalize $\rho$ at the given parameter point, so that it has the form $\rho=\sum_j \lambda_j(|j\>\<j|)$. Then we introduce a basis for the space of traceless Hermitian operators following \rcite{GillM00},
		\begin{equation}\label{eq:GMbasis}
		\begin{aligned}
		\rho_{, jk+}=&|j\>\<k|+|k\>\<j|,\quad \rho_{, jk-}=-\rmi(|j\>\<k|-|k\>\<j|), \quad j<k,\\
		\rho_{,m}=&\sum_j c_{mj} |j\>\<j|,\quad m=1,2,\ldots, d-1,
		\end{aligned}
		\end{equation}
		where the real coefficients $c_{mj}$ satisfy
		\begin{equation}
		\sum_j c_{mj}=0,\quad \sum_j \frac{1}{\lambda_j}c_{m' j}c_{mj}=\delta_{m'm},
		\end{equation}
		which implies that
		\begin{equation}\label{eq:cmjIdentity}
		\sum_m c_{mj}c_{mk}=\lambda_j\delta_{jk}-\lambda_j\lambda_k. 
		\end{equation}
		The operator basis in \eref{eq:GMbasis} determines an affine parametrization in the state space,
		\begin{equation}
		\rho(\theta)=\sum_j \lambda_j(|j\>\<j|)+\sum_{m=1}^{d-1}\theta_m\rho_{,m} +\sum_{j<k}\bigl(\theta_{jk+}\rho_{,jk+}+\theta_{jk-}\rho_{,jk-}\bigr).
		\end{equation}	
		The associated symmetric logarithmic derivatives read 
		\begin{equation}
		L_{jk+}=\frac{2\rho_{, jk+}}{\lambda_j+\lambda_k},\quad L_{jk-}=\frac{2\rho_{, jk-}}{\lambda_j+\lambda_k},\quad
		L_{m}=\sum_j \frac{c_{mj}}{\lambda_j} |j\>\<j|.
		\end{equation}	 	 
		With this parametrization, the quantum Fisher information matrix is diagonal at the given parameter point, with diagonal entries given by 
		\begin{equation}
		J_{jk\pm, jk\pm}=\frac{4}{\lambda_j+\lambda_k}, \quad  J_{m,m}=1.
		\end{equation}

		Now we are ready to prove \lref{lem:QFIsumMix}.
		\begin{align}
		&\sum_{a,b=1}^{d^2-1} (J^{-1})_{ab}(\rho_{,a}\otimes
		\rho_{,b})=\sum_{j<k}\frac{\lambda_j+\lambda_k}{4}\bigl[(|j\>\<k|+|k\>\<j|)^{\otimes 2}-(|j\>\<k|-|k\>\<j|)^{\otimes 2}\bigr]+\sum_m \biggl(\sum_j c_{mj}|j\>\<j|\biggr)^{\!\!\otimes 2}\nonumber\\
		&=\sum_{j< k}\frac{\lambda_j+\lambda_k}{2}(|jk\>\<kj|+|kj\>\<jk|)+\sum_{j,k}\Biggl[\biggl(\sum_m c_{mj} c_{mk}\biggr)(|jk\>\<jk|)\Biggr]\nonumber\\
		&=\sum_{j\neq k}\frac{\lambda_j+\lambda_k}{2}(|jk\>\<kj|)+\sum_{j,k}(\lambda_j\delta_{jk}-\lambda_j\lambda_k)(|jk\>\<jk|)\nonumber\\
		&=\sum_{j, k}\frac{\lambda_j+\lambda_k}{2}(|jk\>\<kj|)-\sum_{j,k}\lambda_j\lambda_k(|jk\>\<jk|)=\frac{1}{2}V(\rho\otimes 1+1\otimes
		\rho)-\rho^{\otimes2}.
		\end{align}
		Here the third equality follows from \eref{eq:cmjIdentity}.
	\end{proof}

	\subsection{Gill-Massar bound}
	In this section we provide more details on the Gill-Massar bound for the scaled WMSE \cite{GillM00,Zhu12the,HouZXL16}. Let $C^{(N)}(\theta)$ be the MSE matrix of any unbiased estimator, then $[C^{(N)}(\theta)]^{-1}\leq I^{(N)}(\theta)$ according to the Cram\'er-Rao bound \cite{Rao02book}. In conjunction with the GM inequality in \eref{eq:GMineqSupp}, we deduce
	\begin{equation}
	\tr\bigl ( J^{-1}(\theta)[C^{(N)}(\theta)]^{-1}\bigr) \leq N(d-1). 
	\end{equation} 
	This inequality imposes a lower bound for the scaled WMSE $N\tr(WC^{(N)})$ for any positive weighting matrix $W$~\cite{GillM00, Zhu12the} (to simplify the notation  the dependence on the parameter point $\theta$ is suppressed),
	\begin{equation}\label{eq:GMboundWMSE}
	N\tr(WC^{(N)})\geq \mse_{W}^{\mathrm{GM}}:=\frac{\bigl(\tr\sqrt{J^{-1/2}WJ^{-1/2}}\,\bigr)^2}{d-1}. 
	\end{equation}
	This bound applies to arbitrary measurements on $\caH^{\otimes N}$ for pure states and to separable measurements for mixed states. In the case $N=1$ and $d=2$, this bound was first derived by Hayashi~\cite{Haya97inbook}.

	When independent and  identical measurements
	are performed on individual copies of $\rho$, we have $I^{(N)}(\theta)=NI(\theta)$, where $I(\theta)$ is the Fisher information matrix associated with the measurement on $\caH$. In addition, the bound $[C^{(N)}(\theta)]^{-1}\leq I^{(N)}(\theta)$ can be saturated by the maximum-likelihood estimator in the large-$N$ limit \cite{Rao02book,PariR04Book} (saturated approximately when $N$ is reasonably large). Then the lower bound in \eref{eq:GMboundWMSE} is saturated iff the measurement on $\caH$ yields 
	the Fisher information matrix \cite{Zhu12the}
	\begin{equation}\label{eq:FisherGM}
	I_W=(d-1)J^{1/2}\frac{\sqrt{J^{-1/2}WJ^{-1/2}}}{\tr\sqrt{J^{-1/2}WJ^{-1/2}}}J^{1/2}.
	\end{equation}

	The Bures distance between two quantum states $\rho$ and $\sigma$ is defined as $D_\rmB(\rho,\sigma)=\sqrt{2-2\sqrt{F(\rho,\sigma)}}$, where $F(\rho,\sigma)=(\tr\sqrt{\rho^{1/2}\sigma\rho^{1/2}}\bigr)^2$ is the fidelity \cite{BengZ06book}. The
	quantum Fisher information matrix $J$ allows  defining a statistical distance in the state space that  is equal to four times of  the infinitesimal Bures distance~\cite{BrauC94}.
	So the weighting matrix for the mean square Bures distance (MSB) is equal to $J/4$, and  the corresponding GM bound is \cite{Zhu12the}.
	\begin{equation}
	\msb^{\mathrm{GM}}=\begin{cases}
	\frac{(d+1)^2(d-1)}{4} &\mbox{mixed states,}\\
	d-1 &\mbox{pure states.}
	\end{cases}
	\end{equation}
	A measurement saturates this bound iff it yields the Fisher information matrix $I=J/(d+1)$ for mixed states or $I=J/2$ for pure states, in which case  the measurement is Fisher symmetric. Note that the infinitesimal infidelity $1-F$ is equal to the infinitesimal square Bures distance, so similar conclusion holds if the MSB is replaced by the mean infidelity when $N$ is sufficiently large.

	For pure states, the infinitesimal square Hilbert-Schmidt distance is equal to two times of the infinitesimal square Bures distance. So the GM bound for the scaled MSE (with respect to  the Hilbert-Schmidt distance) is $2(d-1)$, which is also saturated by Fisher-symmetric measurements.

	Next, we discuss briefly  the derivation of \crref{cor:WMSEent2} from \thref{thm:GMT2} in the main text. To this end, it is instructive
	to further clarify the problem setting; cf.~\sref{sec:QSTcoll} in this supplement. Here we assume that  $N$ is even and $N$ identically prepared quantum systems available for tomography are divided into $N/2$ pairs. Each time  we can measure a pair of quantum systems together;  in other words,
	we can perform collective measurements on $\caH^{\otimes2}$. The simplest strategy is to repeat a given collective measurement $N/2$ times. In general, adaptive measurements are also allowed; that is, the choices of later measurements may depend on the outcomes of previous measurements. Nevertheless,  each measurement setting is usually repeated  many times to acquire reliable statistics.

	Given the above explanation, it is clear that  the derivation of \crref{cor:WMSEent2} from \thref{thm:GMT2} in the main text follows from a similar reasoning  employed for deriving the GM bound 
	as outlined above. The conclusion is also similar except that 
	the bound for the scaled WMSE is reduced by a factor of $2/3$ because the bound for $\tr(J^{-1}I^{(2)})$ is increased by 50\%.

	\section{Proof of \thref{thm:UFSnogo}}
	In this section we prove \thref{thm:UFSnogo}, which states that no single-copy  measurement with a finite number of outcomes is Fisher symmetric for all pure states. 
	\begin{proof}
		Suppose the POVM $\{\Pi_\xi\}$ on $\caH$ is Fisher symmetric for all pure states. Then all $\Pi_\xi$ have rank one according to \pref{pro:GMTPure}. Let $\rho$ be a pure state that has orthogonal support with $\Pi_\xi$, that is, $\tr(\rho\Pi_\xi)=0$. Then the Fisher information matrix at $\rho$ provided by $\{\Pi_\xi\}$ cannot saturate the GM inequality again according to \pref{pro:GMTPure}. Therefore, $\{\Pi_\xi\}$ is not Fisher symmetric at $\rho$. This contradiction confirms \thref{thm:UFSnogo}. 
	\end{proof}
	Note that this proof also applies to POVMs with countable outcomes, but not to continuous POVMs. For a continuous POVM $\{\Pi_\xi\}$, the requirement $\tr(\rho\Pi_\xi)> 0$ stated in \pref{pro:GMTPure} may be violated by a subset of outcomes that has measure zero. That is why \thref{thm:UFSnogo} does not contradict the fact that the covariant measurement is Fisher symmetric for all pure states. 
	
	In view of \thref{thm:UFSnogo}, it is natural to ask whether there exists a finite POVM on $\caH$ that is Fisher symmetric for all pure states, except for a set of measure zero. We believe that such POVMs cannot exist, but have not found a rigorous proof so far. As an evidence in support of our belief,
	in the case of a qubit, calculation shows  that POVMs constructed from platonic solids inscribed in the Bloch sphere are generally not Fisher symmetric except for certain special  points on the Bloch sphere (depending on the platonic solids chosen). Nevertheless, it is possible to construct POVMs that are Fisher symmetric for almost all pure states on any given great circle on the Bloch sphere. For example, suppose  pure qubit states are parameterized as 
	\begin{equation}
	\rho(\theta,\phi)=|\psi(\theta,\phi)\>\<\psi(\theta,\phi)|,\quad |\psi(\theta,\phi)\>=\cos\Bigl(\frac{\theta}{2}\Bigr)|0\>+\sin\Bigl(\frac{\theta}{2}\Bigr)\rme^{\rmi\phi}|1\>, \quad 0\leq \theta\leq \pi,\quad 0\leq \phi\leq 2\pi.
	\end{equation}
	Then the quantum Fisher information matrix is given by 
	\begin{equation}
	J(\theta,\phi)=\diag(1,\sin^2(\theta)).
	\end{equation}
	Consider the POVM $\{\Pi_j\}$ composed of the four elements
	\begin{equation}
	\Pi_1=\frac{1}{2}(|0\>\<0|), \quad \Pi_2=\frac{1}{2}(|1\>\<1|), \quad\Pi_3=\frac{1}{2}(|+\>\<+|), \quad
	\Pi_4=\frac{1}{2}(|-\>\<-|),
	\end{equation}
	where $|\pm\>=\frac{1}{\sqrt{2}}(|0\>\pm|1\>)$. 
	The probability of obtaining the four outcomes are respectively given by
	\begin{equation}
	p_1=\frac{1+\cos(\theta)}{2},\quad p_2=\frac{1-\cos(\theta)}{2}, \quad p_3=\frac{1+\sin(\theta)\cos(\phi)}{2},\quad p_4=\frac{1-\sin(\theta)\cos(\phi)}{2}.
	\end{equation}
	The Fisher information matrix provided by the POVM reads
	\begin{equation}
	I(\theta,\phi)=\frac{1}{2}
	\begin{pmatrix}
	1+\frac{\cos^2(\theta)\cos^2(\phi)}{1-\sin^2(\theta)\cos^2(\phi)} &-\frac{\sin(\theta)\cos(\theta)\sin(\phi)\cos(\phi)}{1-\sin^2(\theta)\cos^2(\phi)}\\
	-\frac{\sin(\theta)\cos(\theta)\sin(\phi)\cos(\phi)}{1-\sin^2(\theta)\cos^2(\phi)} &\frac{\sin^2(\theta)\sin^2(\phi)}{1-\sin^2(\theta)\cos^2(\phi)}
	\end{pmatrix}.
	\end{equation}
	The POVM $\{\Pi_j\}$ is Fisher symmetric for the the parameter point $(\theta, \phi)$ if $p_j\neq0$ for $j=1,2,3,4$ and $I(\theta,\phi)=\frac{1}{2}J(\theta,\phi)=\frac{1}{2}\diag(1,\sin^2(\theta))$. These conditions hold iff $\theta=\pi/2$ and $\phi\neq 0,\pi$, or $\phi=\pi/2, 3\pi/2$ and $\theta\neq 0, \pi$. These parameter points form the union of two great circles with four points corresponding to $\Pi_1, \Pi_2, \Pi_3,\Pi_4$  deleted. Incidentally, the POVM $\{|0\>\<0|, |1\>\<1|\}$ saturates the quantum Cram\'er-Rao bound for $\theta$ whenever $\theta\neq 0,\pi$; the POVM $\{|+\>\<+|, |-\>\<-|\}$ saturates the quantum Cram\'er-Rao bound for $\phi$ on the great circle with $\theta=\pi/2$ as long as $\phi\neq 0,\pi$; cf. \rcite{BarnG00}. 
	
	\section{Generalized 2-designs  and generalized SICs}
	Here we provide additional details on generalized 2-designs and generalized SICs  \cite{Zaun11,ReneBSC04,Appl07,GourK14,Zhu14T,GrayA16}, which are relevant to the current study. Most results presented here are more or less known before, but some of them have not been stated explicitly or clearly. Note that generalized 2-designs considered in this paper are slightly  different from conical designs studied in \rcite{GrayA16}, although they share a common spirit.

	\subsection{Generalized 2-designs}
	A weighted set of quantum states $\{|\psi_\xi\>, w_\xi\}$   in dimension $d$ is  a (weighted complex projective) 2-design \cite{Zaun11,ReneBSC04,Scot06}  if 
	\begin{equation}\label{eq:2designSupp}
	\sum_{\xi} w_\xi (|\psi_\xi\>\<\psi_\xi|)^{\otimes 2 }= \frac{2\sum_{\xi}w_\xi}{d(d+1)} P_+, 
	\end{equation} 
	where $P_+$ ($P_-$) is the projector onto the symmetric  subspace $\caH_+$ (antisymmetric  subspace $\caH_-$)  of $\caH^{\otimes 2}$. 
	It is straightforward to verify that 
	\begin{equation}
	\sum_{\xi, \eta } w_\xi w_\eta  |\<\psi_\xi|\psi_\eta \>|^4\geq  \frac{2}{d(d+1)}\biggl(\sum_{\xi}w_\xi\biggr)^2, 
	\end{equation}
	and the lower bound is saturated iff $\{|\psi_\xi\>, w_\xi\}$ is a 2-design. 
	Any 2-design $\{|\psi_\xi\>, w_\xi\}$ has at least $d^2$ elements, and the lower bound is saturated iff all weights $w_\xi$ are equal, and $|\psi_\xi\>$ form a symmetric informationally complete measurement \cite{Zaun11,ReneBSC04, Scot06,ApplFZ15G,Zhu14T} (SIC for short), that is, 
	\begin{equation}
	|\<\psi_\xi|\psi_\eta\>|^2=\frac{d\delta_{\xi\eta}+1}{d+1},\quad \forall \xi, \eta. 
	\end{equation}
	Another prominent example of 2-designs are complete sets of mutually unbiased bases (MUB) \cite{WootF89,DurtEBZ10,Zhu15M}. In the case of a qubit, every SIC defines a regular tetrahedron on the Bloch sphere, and vice versa. By contrast, every complete set of MUB defines a regular octahedron on the Bloch sphere, and vice versa. For example, one SIC is composed of the four states with Bloch vectors 
	\begin{equation}\label{eq:SICqubitBloch}
	\frac{1}{\sqrt{3}}(1,1,1),\quad \frac{1}{\sqrt{3}}(1,-1,-1),\quad \frac{1}{\sqrt{3}}(-1,1,-1),\quad \frac{1}{\sqrt{3}}(-1,-1,1),
	\end{equation}
	respectively. One complete set of MUB is composed of the six states with Bloch vectors $(\pm1,0,0), (0,\pm1,0), (0,0,\pm1)$.

	A set of positive operators $\{\Pi_\xi\}$ is a generalized 2-design (cf. conical designs in \rcite{GrayA16}) if 
	\begin{align}\label{eq:g2designSupp}
	\sum_{\xi}\frac{ \Pi_\xi \otimes \Pi_\xi}{\tr(\Pi_\xi)}=&\frac{\sum_{\xi} w_\xi}{d}\Bigl(\frac{1+\wp}{d+1}P_++\frac{1-\wp}{d-1}P_-\Bigr), 
	\end{align}
	where 
	\begin{equation}
	w_\xi:=\tr(\Pi_\xi), \quad 
	\wp:= \frac{\sum_\xi w_\xi \wp_\xi}{\sum_{\xi} w_\xi}, \quad \wp_\xi:= \frac{\tr(\Pi_\xi^2)}{[\tr(\Pi_\xi)]^2}.
	\end{equation}
	Here $\wp_\xi$ may be interpreted as the purity of $\Pi_\xi$, recall that the purity of a positive operator $A$ is defined as $\tr(A^2)/[\tr(A)]^2$. 
	Therefore, $\wp$ is the (weighted) average purity of $\{\Pi_\xi\}$. Note that $1/d\leq \wp\leq 1$. The lower bound is saturated iff  all $\Pi_\xi$ are proportional to the identity, in which case $\{\Pi_\xi\}$ forms a trivial generalized 2-design. 
	We will assume $\wp>1/d$ except when stated otherwise.  
	The upper bound $\wp\leq 1$ is saturated   iff all  $\Pi_\xi$ are rank one and thus can be expressed in the form $\Pi_\xi=w_\xi (|\psi_\xi\>\<\psi_\xi|)$, in which case  \eref{eq:g2designSupp} reduces to \eref{eq:2designSupp}. So a generalized 2-design of purity~1 is a 2-design, and vice versa.  Many generalized 2-designs of lower purities can also be constructed from 2-designs. For example, let $\{\rho_\xi\}$ be the SIC defined by the Bloch vectors in \eref{eq:SICqubitBloch}, then $\{\rho_\xi+a\}$ is a generalized 2-design for any positive constant $a$.

	Taking the partial trace in \eref{eq:g2designSupp} yields
	\begin{equation}
	\sum_\xi\Pi_\xi= \frac{\sum_{\xi}w_\xi}{d},
	\end{equation} 
	which implies the following proposition. 
	\begin{proposition}
		Any generalized 2-design forms a POVM after proper rescaling. 
	\end{proposition}

	In the rest of this section, we assume that
	$\{\Pi_\xi\}$ is a set of nonzero positive operators in dimension~$d$ with  normalization  $\sum_{\xi}\tr(\Pi_\xi)=d$. Under this convention,  $\{\Pi_\xi\}$ forms a POVM whenever it is a generalized 2-design. 
	The general case can be analyzed  by proper rescaling.

	The following proposition provides a simple characterization of  generalized 2-designs. 
	\begin{proposition}\label{pro:g2designBound}
		Suppose $\{\Pi_\xi\}$ has purity $\wp$. Then 
		\begin{align}
		\sum_{\xi,\eta}\frac{[\tr (\Pi_\xi\Pi_\eta)]^2}{\tr(\Pi_\xi)\tr(\Pi_\eta)}
		\geq &\frac{d^2(1+\wp^2)-2d\wp}{d^2-1}, \label{eq:g2designBound}
		\end{align}
		and the lower bound is saturated iff $\{\Pi_\xi\}$ is a generalized 2-design. 
	\end{proposition}

	\begin{proof} 	Let 
		\begin{align}
		\caM:=\sum_{\xi}\frac{ \Pi_\xi \otimes \Pi_\xi}{\tr (\Pi_\xi)},\quad \caM_\pm:=P_\pm\caM P_\pm.
		\end{align}
		Then
		\begin{equation}
		\tr(\caM_+ )=\frac{1}{2}\sum_{\xi}\frac{ [\tr(\Pi_\xi)]^2+\tr(\Pi_\xi^2)}{\tr(\Pi_\xi)}=\frac{1}{2}d(1+\wp),\quad  \tr(\caM_- )=\frac{1}{2}\sum_{\xi}\frac{ [\tr(\Pi_\xi)]^2-\tr(\Pi_\xi^2)}{\tr(\Pi_\xi)}=\frac{1}{2}d(1-\wp). 
		\end{equation}
		Therefore, 
		\begin{align}
		\tr(\caM^2)\geq& \tr(\caM_+^2)+\tr(\caM_-^2)\geq \Bigl(\frac{1+\wp}{d+1}\Bigr)^2\tr(P_+^2)+\Bigl(\frac{1-\wp}{d-1}\Bigr)^2\tr(P_-^2)=
		\frac{d(1+\wp)^2}{2(d+1)}+\frac{d(1-\wp)^2}{2(d-1)}\nonumber\\
		=&\frac{d^2(1+\wp^2)-2d\wp}{d^2-1}. 
		\end{align}
		Here the first inequality is saturated iff $\caM=\caM_++\caM_-$, and the second one is saturated iff $\caM_+$ and $\caM_-$ are proportional to $P_+$ and $P_-$, respectively. The two inequalities are saturated simultaneously iff 
		\eref{eq:g2designSupp} with $\sum_\xi w_\xi=d$ is satisfied, in which case  $\{\Pi_\xi\}$ is a generalized 2-design. 	
	\end{proof}

	Next, we clarify the connection between  generalized  2-designs and  tight IC measurements introduced by Scott \cite{Scot06}, which is helpful to studying (weakly) Fisher-symmetric measurements. To this end, we need to introduce the formalism of superoperators following \rscite{ZhuE11,Zhu14IOC,ApplFZ15G}. With respect to the Hilbert-Schmidt inner product, the set of operators on $\caH$ forms a Hilbert space. The vectors in this space are denoted using the double-ket notation; for example, the operator $A$ is denoted by $\dket{A}$.  The inner product between $A$ and $B$ is written as $\dinner{A}{B}:=\tr(A^\dag B)$. Operators acting on this space, so called superoperators, can be constructed from  outer products, such as $\dket{A}\dbra{B}$,  which induces the linear transformation $\dket{C}\to\dket{A}\dinner{B}{C}$. The identity superoperator is denoted by $\mathbf{I}$.

	A POVM  $\{\Pi_\xi\}$ on $\caH$ is \emph{tight IC}  \cite{Scot06,ZhuE11,Zhu14T} if it satisfies the following equation
	\begin{align}
	\sum_{\xi}\frac{ \douter{\Pi_\xi}{\Pi_\xi}}{\tr(\Pi_\xi)}=&\frac{d\wp-1}{d^2-1}\mathbf{I}+\frac{d-\wp}{d^2-1}\douter{1}{1}, \label{eq:TightIC}
	\end{align}
	where $\wp$ is the purity of $\{\Pi_\xi\}$. 
	Note that the requirement  $\sum_{\xi }\Pi_\xi=1$ for a POVM is satisfied automatically if  \eref{eq:TightIC} holds. According to Lemma~1 in \rcite{ApplFZ15G}, \eref{eq:TightIC} is equivalent to \eref{eq:g2designSupp} with $\sum_{\xi}w_\xi=d$.  This observation leads to the following proposition. 
	\begin{proposition}\label{pro:g2designTIC}
		A POVM $\{\Pi_\xi\}$ is tight IC iff it is a generalized 2-design.
	\end{proposition}
	When $\wp>1/d$, the superoperator in \eref{eq:TightIC} has full rank. It follows that any nontrivial generalized 2-design $\{\Pi_\xi\}$ spans the whole operator space and thus
	has at least $d^2$ elements.  Generalized 2-designs saturating  the lower bound are called \emph{minimal}. They are 
	characterized by  the following proposition, which  follows from Theorem~1 in \rcite{ApplFZ15G} and Lemma~1 in \rcite{Zhu14T}. 
	\begin{proposition}\label{pro:g2designMin}
		Any nontrivial generalized 2-design 
		has at least  $d^2$ elements.  The set $\{\Pi_\xi\}$ is a minimal generalized 2-design of purity $\wp>1/d$ iff
		\begin{align}
		\tr(\Pi_\xi\Pi_\eta)=
		\frac{d\wp-1}{d^2-1}\sqrt{\tr(\Pi_\xi)\tr(\Pi_\eta)}\delta_{\xi\eta}+\frac{d-\wp}{d^2-1}\tr(\Pi_\xi)\tr(\Pi_\eta)\quad \forall \xi, \eta. \label{eq:g2designMin}
		\end{align}
	\end{proposition}
	\begin{remark}
		If  \eref{eq:g2designMin} holds, then  $\{\Pi_\xi\}$ spans the operator space and  $\sum_{\xi}\Pi_\xi=1$. To see this, let
		$L_\xi:=\Pi_\xi/\sqrt{\tr(\Pi_\xi)}$, then  $\tr(L_\xi)=\sqrt{\tr(\Pi_\xi)}$ and 
		\begin{equation}
		\tr(L_\xi L_\eta) =\frac{d\wp-1}{d^2-1} \delta_{\xi\eta} +\frac{d-\wp}{d^2-1} \tr(L_\xi)\tr(L_\eta). 
		\end{equation}
		Observing that the Gram matrix of $\{L_\xi\}$ is positive definite whenever $\wp>1/d$, we conclude that $L_\xi$ are linearly independent and span the whole operator space and that the same holds for $\{\Pi_\xi\}$. Summing over $\xi$ in \eref{eq:g2designMin}, we can deduce that $\sum_{\xi}\Pi_\xi$ is proportional to the identity. Then it is straightforward to show that $\sum_{\xi}\tr(\Pi_\xi)=d$, which implies that   $\sum_{\xi}\Pi_\xi=1$. 
	\end{remark}

	\subsection{Generalized SIC}
	A POVM $\{\Pi_\xi\}$ is a generalized SIC \cite{Appl07,GourK14,Zhu14T} if it has $d^2$ elements and satisfies
	\begin{equation}\label{eq:gSICsupp}
	\tr (\Pi_\xi\Pi_\eta)=\alpha \delta_{\xi\eta}+\beta \quad \forall \xi,\eta
	\end{equation}
	for some positive constants $\alpha,\beta$. Together with the requirement $\sum_{\xi}\Pi_\xi=1$, \eref{eq:gSICsupp} implies that 
	\begin{equation}
	\tr(\Pi_\xi)=\alpha+d^2\beta=\frac{1}{d},\quad  \tr(\Pi_\xi^2)=\alpha+\beta\quad \forall \xi.
	\end{equation}
	Therefore, all $\Pi_\xi$ have the same trace of $1/d$ and the same purity of $\wp=d^2(\alpha+\beta)$.  The constants $\alpha$ and $\beta$ are determined by the purity as
	\begin{equation}
	\alpha=\frac{d\wp-1}{d(d^2-1)},\quad \beta=\frac{d-\wp}{d^2(d^2-1)}.
	\end{equation}
	In this paper,   $\{c\Pi_\xi\}$ for any positive constant $c$ is also called a generalized SIC whenever $\{\Pi_\xi\}$ is. 
	By virtue of \pref{pro:g2designBound} or \ref{pro:g2designMin}, it is straightforward to verify that any generalized SIC is a generalized 2-design.
	Conversely, if   $\{\Pi_\xi\}$ is a minimal  generalized 2-design with
	all $\Pi_\xi$ having the same purity, then \eref{eq:g2designMin} implies that all $\Pi_\xi$ have the same trace of $1/d$, so that $\{\Pi_\xi\}$ is a generalized SIC (cf.~\rcite{Zhu14T}), which satisfies
	\begin{align}
	\tr (\Pi_\xi\Pi_\eta)=\frac{d\wp-1}{d(d^2-1)}\delta_{\xi\eta}+\frac{d-\wp}{d^2(d^2-1)}\quad \forall \xi, \eta. 
	\end{align}  
	\begin{proposition}\label{pro:gSICg2design}
		Any minimal generalized 2-design whose  elements have the same purity
		is 	a generalized SIC, and vice versa. 	 	
	\end{proposition}
	\Pref{pro:gSICg2design} implies the following corollary proved in \rscite{ReneBSC04,Scot06}; see also \rscite{Zhu14T,ApplFZ15G}.
	\begin{corollary}\label{cor:SIC2design}
		Any minimal 2-design is 	a  SIC, and vice versa. 	 	
	\end{corollary}

	\subsection{Construction of generalized 2-designs from unitary 2-designs}

	In this section we introduce a method for constructing generalized 2-designs from unitary 2-designs. This method allows us to construct generalized 2-designs all of whose elements are proportional to  projectors of a given rank, which are useful in constructing tight coherent measurements introduced in the main text. 
	
	A set of weighted  unitary operators $\{U_\xi,w_\xi\}$ on $\caH$ is a (weighted)  \emph{unitary $t$-design} \cite{DankCEL09, GrosAE07} if it satisfies
	\begin{equation}\label{eq:U2design}
	\sum _\xi w_\xi U_\xi^{\otimes t} M(U_\xi^{\otimes t})^\dag =\sum_{\xi }w_\xi\int \rmd U U^{\otimes t}M(U^{\otimes t})^\dag
	\end{equation}
	for any operator $M$ acting on $\caH^{\otimes t}$.  Here the symbol $\dag$ denotes the Hermitian conjugate, and
	the integral is taken  with respect to the normalized Haar measure over the whole unitary group. It is known that  a unitary $t$-design exists  for any  positive integer $t$, assuming $\caH$ has a finite dimension \cite{RoyS09}. 
	
	Suppose $\{U_\xi,w_\xi\}$ is a unitary 2-design, and $\Pi$ is a positive operator. Let $\Pi_\xi:=w_\xi U_\xi\Pi U_\xi^\dag$, then $\{\Pi_\xi\}$ is a generalized 2-design. To see this, note that 
	\begin{align}
	\sum_{\xi}\frac{ \Pi_\xi \otimes \Pi_\xi}{\tr(\Pi_\xi)}=&\sum_{\xi}\frac{w_\xi(U_\xi\Pi U_\xi^\dag) \otimes (U_\xi\Pi U_\xi^\dag)}{\tr(\Pi)}=\frac{1}{\tr(\Pi)}\sum_{\xi}w_\xi U_\xi^{\otimes 2}\Pi^{\otimes 2} \bigl(U_\xi^{\otimes 2}\bigr)^\dag\nonumber\\
	=&\frac{\sum_\xi w_\xi}{\tr(\Pi)}\int \rmd U U^{\otimes 2}\Pi^{\otimes 2}(U^{\otimes 2})^\dag
	=\alpha P_+ +\beta P_-,
	\end{align}
	where $\alpha$ and $\beta$ are nonnegative constants.
	Here the last equality follows from the fact that the second tensor power of the unitary group on $\caH$ has two inequivalent irreducible components on $\caH^{\otimes2}$, which correspond to the symmetric and antisymmetric subspaces, respectively. This observation shows that  $\{\Pi_\xi\}$ is indeed a generalized 2-design. When $\Pi$ is a pure state, we get a 2-design; when $\Pi$ is a rank-$k$ projector, we get a generalized 2-design all of whose elements are proportional to rank-$k$ projectors.

	\section{Fisher-symmetric measurements at the completely mixed state} 
	In this section we prove \pref{pro:WFSmix} in the main text, which states that a POVM $\{\Pi_\xi\}$ is (weakly) Fisher symmetric at the completely mixed state iff $\{\Pi_\xi\}$ is a (generalized) 2-design. In view of  \pref{pro:g2designTIC},  it suffices to show that a POVM  is weakly Fisher symmetric at the completely mixed state iff it is tight IC. \Crref{cor:FSmixMinOne} is an immediate consequence of  \pref{pro:WFSmix} and \crref{cor:SIC2design}.

	\begin{proof}[Proof of \pref{pro:WFSmix}]
		Let $\{E_a\}_{a=1}^{d^2-1}$ be an orthonormal basis (with respect to the Hilbert-Schmidt inner product) for the space of traceless Hermitian operators. Choose the following affine parametrization on the state space
		\begin{equation}\label{eq:AffinePara}
		\rho(\theta)=\frac{1}{d}+\sum_a \theta_a E_a,
		\end{equation}	
		then $\rho_{,a}=\partial\rho/\partial\theta_a=E_a$. At the completely mixed state, the symmetric logarithmic derivatives read $L_a=d\rho_{,a}=dE_a$, so the quantum Fisher information matrix  has entries
		\begin{equation}
		J_{ab}=d\tr(E_aE_b) =d \delta_{ab}. 
		\end{equation}
		In addition, the Fisher information matrix has entries
		\begin{equation}\label{eq:FisherMixOne}
		I_{ab}=d\sum_{\xi}\frac{ \tr(E_a  \Pi_\xi) \tr(E_b\Pi_\xi)}{\tr (\Pi_\xi)} =d\dbra{E_a}\caF\dket{E_b},
		\end{equation}
		where the superoperator $\caF$ is defined as 
		\begin{equation}
		\caF:=\sum_{\xi}\frac{ \douter{\Pi_\xi}{\Pi_\xi}}{\tr (\Pi_\xi)}. 
		\end{equation}	
		The  POVM is weakly Fisher symmetric at the completely mixed state iff $I$ is proportional to the identity, which holds iff the superoperator $\caF$ has the form of \eref{eq:TightIC}, in which case $\{\Pi_\xi\}$ is a tight IC POVM and a generalized 2-design. The  POVM is  Fisher symmetric iff $I=J/(d+1)$, which holds iff the superoperator $\caF$ has the form of \eref{eq:TightIC} with $\wp=1$.  In that  case, $\{\Pi_\xi\}$ is a 2-design, and vice versa. 
	\end{proof}

	\Pref{pro:WFSmix} can also be proved without resorting to superoperators. Observing that 
	\begin{equation}
	\sum_a \tr(E_a\Pi_\xi)\tr(E_a\Pi_\eta)=\tr(\Pi_\xi\Pi_\eta)-\frac{\tr(\Pi_\xi)\tr(\Pi_\eta)}{d},
	\end{equation}
	we deduce that
	\begin{align}\label{eq:FsiherTrace}
	\tr (I)=\sum_{\xi}\frac{d\tr (\Pi_\xi^2)- [\tr(\Pi_\xi)]^2}{\tr(\Pi_\xi)}=d^2\wp-d, 
	\end{align}
	where $\wp$ is the purity of $\{\Pi_\xi\}$. In addition, 
	\begin{align}\label{eq:FsiherTraceSq}
	\tr (I^2)=&\sum_{a,b}I_{ab}^2= \sum_{a,b}\sum_{\xi,\eta}d^2\frac{ \tr(E_a  \Pi_\xi) \tr(E_b\Pi_\xi)}{\tr (\Pi_\xi)}\frac{ \tr(E_a  \Pi_\eta) \tr(E_b\Pi_\eta)}{\tr (\Pi_\eta)}\nonumber\\
	=&
	\sum_{\xi,\eta}\frac{[d\tr (\Pi_\xi\Pi_\eta)- \tr(\Pi_\xi)\tr(\Pi_\eta)]^2}{\tr(\Pi_\xi)\tr(\Pi_\eta)}=d^2\sum_{\xi,\eta}\frac{[\tr (\Pi_\xi\Pi_\eta)]^2}{\tr(\Pi_\xi)\tr(\Pi_\eta)}-d^2. 
	\end{align}
	The POVM $\{\Pi_\xi\}$ is weakly Fisher symmetric  iff $\tr(I^2)=[\tr(I)]^2/(d^2-1)$, which holds iff the inequality in \eref{eq:g2designBound} is saturated, in which case $\{\Pi_\xi\}$ is a generalized 2-design. The POVM is  Fisher symmetric if in addition $\tr(I)=\tr(J)/(d+1)=d^2-d$, which demands that $\wp=1$, so that  $\{\Pi_\xi\}$ is a 2-design; the converse is also immediate.
	
	\section{Proof of \thref{thm:GMT2}}
	In this section we prove the inequality  $\tr [J^{-1}(\theta)I^{(2)} (\theta)]\leq 3d-3$  in \thref{thm:GMT2} and determine the condition for saturating the inequality. 
	The simple idea can be explained as follows: the optimal POVM is always the union of two POVMs on the symmetric and antisymmetric subspaces, respectively; the value of $\tr(J^{-1} I^{(2)})$ is maximized iff the marginal of each POVM element has the highest purity under the given symmetry. To simplify the notation, the parameter point $\theta$ is omitted.

	\begin{proof}[Proof of \thref{thm:GMT2}]
		The Fisher information matrix $I^{(2)}$ provided by $\{\Pi_\xi\}$ has matrix elements
		\begin{align}
		I^{(2)}_{ab}=&\frac{\tr[(\rho\otimes \rho_{,a}+\rho_{,a}\otimes \rho)\Pi_\xi]\tr[(\rho\otimes \rho_{,b}+\rho_{,b}\otimes \rho)\Pi_\xi]}{\tr(\rho^{\otimes 2}\Pi_\xi)}
		=2\sum_{\xi} \frac{ \tr(\tilde{E}_a  \tilde{Q}_\xi)\tr(\tilde{E}_b \tilde{Q}_\xi)}{\tr(\tilde{Q}_\xi )}\nonumber\\
		=&2\sum_{\xi} \frac{ \tr\bigl[(\tilde{E}_a \otimes\tilde{E}_b) \tilde{Q}_\xi^{\otimes 2}\bigr] }{\tr(\tilde{Q}_\xi )},	
		\end{align}
		where 
		\begin{equation}\label{eq:Qtilde}
		\tilde{E}_a:=\rho^{-1/2}\rho_{,a} \rho^{-1/2}, \quad 
		\tilde{Q}_\xi: = \tr_1(\tilde{\Pi}_\xi)+\tr_2(\tilde{\Pi}_\xi), \quad \tilde{\Pi}_\xi:=(\rho^{1/2})^{\otimes 2}\Pi_\xi(\rho^{1/2})^{\otimes 2}. 
		\end{equation}
		According to \lref{lem:QFIsumMix}, 
		\begin{align}
		&\sum_{a,b=1}^{d^2-1}(J^{-1})_{ab}(\tilde{E}_a\otimes \tilde{E}_b)=\sum_{a,b=1}^{d^2-1}(J^{-1})_{ab}(\rho^{-1/2})^{\otimes 2}(\rho_{,a}\otimes
		\rho_{,b})(\rho^{-1/2})^{\otimes 2}\nonumber\\
		&=(\rho^{-1/2})^{\otimes 2}\Bigl[\frac{1}{2}V(\rho\otimes 1+1\otimes
		\rho)-\rho^{\otimes2}\Bigr](\rho^{-1/2})^{\otimes 2}=\frac{1}{2}V(1\otimes \rho^{-1}+\rho^{-1}\otimes
		1)-1\otimes 1.
		\end{align}
		Therefore, 
		\begin{align}\label{eq:GMTtwo}
		\tr(J^{-1} I^{(2)})=\sum_{a,b=1}^{d^2-1}(J^{-1})_{ab}I^{(2)}_{ab}
		=2\sum_{\xi}\frac{\tr(\rho^{-1}\tilde{Q}_\xi^2)-[\tr(\tilde{Q}_\xi )]^2}{\tr(\tilde{Q}_\xi )}=2\sum_{\xi}\frac{\tr(\rho^{-1}\tilde{Q}_\xi^2)}{\tr(\tilde{Q}_\xi )}-4\leq 3d-3,
		\end{align}
		note that $\sum_{\xi} \tilde{Q}_\xi=2\rho$.  Here the inequality follows from \lref{lem:GMT2} below, which also shows that the inequality is saturated iff each $\Pi_\xi $ is proportional to either the tensor power of a pure state or a Slater-determinant state.
	\end{proof}
	
	\begin{remark}
		The proof of \thref{thm:GMT2} builds on the observation that the value of $\tr(J^{-1} I^{(2)})$ associated with  a POVM is connected to the purities of symmetrized marginals of POVM elements. The proof of \lref{lem:GMT2} below that underpins \thref{thm:GMT2} relies on the fact that the unitary group on the Hilbert space $\mathcal{H}$ has only two irreducible components on  $\mathcal{H}^{\otimes2}$, which correspond to the completely symmetric subspace and completely  antisymmetric subspace, respectively. 
		The  approach adopted  here may also serve as a stepping stone for studying multi-copy collective measurements. However, new ideas are necessary to generalize our proof since the situation becomes more complicated in the multi-copy setting; see \rcite{Zhu12the} for some partial progress along this direction. 
	\end{remark}

	\begin{lemma}\label{lem:GMT2}
		Suppose $\rho$ is a  density matrix  that has full rank  and  $\{\Pi_\xi\} $ is a POVM on $\caH^{\otimes 2}$. Let $\tilde{\Pi}_\xi:=(\rho^{1/2})^{\otimes 2}\Pi_\xi(\rho^{1/2})^{\otimes 2}$ and $\tilde{Q}_\xi: = \tr_1(\tilde{\Pi}_\xi)+\tr_2(\tilde{\Pi}_\xi)$. Then 
		\begin{align}
		\sum_{\xi}\frac{\tr(\rho^{-1}\tilde{Q}_\xi^2)}{\tr(\tilde{Q}_\xi )}\leq \frac{3d+1}{2}, 
		\end{align}
		and the inequality is saturated iff each $\Pi_\xi $ is proportional to either the tensor power of a pure state or a Slater-determinant state.  
	\end{lemma}
	\begin{proof}
		Let $\Pi^{\pm}_\xi:=P_{\pm} \Pi_\xi P_{\pm}$,  $\tilde{\Pi}^{\pm}_\xi:=(\rho^{1/2})^{\otimes 2}\Pi^{\pm}_\xi(\rho^{1/2})^{\otimes 2}
		=P_{\pm} \tilde{\Pi}_\xi P_{\pm}$, and $\tilde{Q}_\xi^{\pm}:= \tr_1(\tilde{\Pi}^{\pm}_\xi)+\tr_2(\tilde{\Pi}^{\pm}_\xi)$;  then $\tilde{Q}_\xi=\tilde{Q}_\xi^+ +\tilde{Q}_\xi^- $. Note that $\tilde{\Pi}_\xi^+ $ and $\tilde{\Pi}_\xi^- $ are supported on the symmetric and antisymmetric subspaces, respectively. Therefore,
		\begin{equation}\label{eq:MarginalSAS}
		(\tilde{Q}_\xi^+ )^2\leq \tr(\tilde{Q}_\xi^+ )\tilde{Q}_\xi^+ ,\quad (\tilde{Q}_\xi^- )^2\leq \frac{1}{2}\tr(\tilde{Q}_\xi^- ) \tilde{Q}_\xi^-,
		\end{equation}
		where the first inequality is saturated iff 
		$\Pi_\xi^+$ (or equivalently $\tilde{\Pi}_\xi^+$) is proportional to  the tensor power of a pure state, and the second one is saturated iff $\Pi_\xi^-$ (or equivalently $\tilde{\Pi}_\xi^-$) is proportional to
		a Slater-determinant state. In addition, we have 
		\begin{align}
		\sum_{\xi}\tr(\rho^{-1}\tilde{Q}^{\pm}_\xi)=\sum_{\xi}\tr\bigl[(\rho^{-1}\otimes 1+1\otimes \rho^{-1})(\rho^{1/2})^{\otimes 2}\Pi^{\pm}_\xi(\rho^{1/2})^{\otimes 2}\bigr]=\tr\bigl[(1\otimes \rho+\rho\otimes 1)P_\pm\bigr]=d\pm 1. 
		\end{align}
		Consequently, 
		\begin{align}
		\sum_{\xi}\frac{\tr(\rho^{-1}\tilde{Q}_\xi^2)}{\tr(\tilde{Q}_\xi )}&\leq 
		\sum_{\xi}\biggl(\frac{\tr\bigl[\rho^{-1}(\tilde{Q}_\xi^+ )^2\bigl]}{\tr(\tilde{Q}_\xi^+  )}+\frac{\tr\bigl[\rho^{-1}(\tilde{Q}_\xi^- )^2\bigl]}{\tr(\tilde{Q}_\xi^-  )}\biggr)
		\leq 
		\sum_{\xi}\tr\bigl(\rho^{-1}\tilde{Q}_\xi^+ \bigl)+\frac{1}{2}\sum_{\xi}\tr\bigl(\rho^{-1}\tilde{Q}_\xi^- \bigl)\nonumber\\
		&=(d+1)+\frac{d-1}{2}
		\leq \frac{3d+1}{2}, 
		\end{align}
		where $\tr\bigl[\rho^{-1}(\tilde{Q}_\xi^\pm )^2\bigl]/\tr(\tilde{Q}_\xi^\pm )$ is set to 0 when $\tilde{Q}_\xi^\pm =0$. Here the first inequality follows from \lref{lem:SumIneq} below, and the second one from \eref{eq:MarginalSAS}.	
		If  $\Pi_\xi $ is proportional to either the tensor power of a pure state or a Slater-determinant state, then the two inequalities are saturated as well as the inequalities in \eref{eq:MarginalSAS}. Conversely, if the two inequalities are saturated, then  $\Pi_\xi^+  $ is proportional to  the tensor power of a pure state, and  $\Pi_\xi^-  $  is proportional to
		a Slater-determinant state. In addition, either $\tilde{Q}_\xi^+ $ or  $\tilde{Q}_\xi^- $
		must vanish in order to saturate the first inequality; that is, either $\Pi_\xi^+ $ or $\Pi_\xi^- $ must vanish. Therefore, $\Pi_\xi $ is proportional to either the tensor power of a pure state or a Slater-determinant state. 
	\end{proof}
	
	\begin{lemma}\label{lem:SumIneq}
		Suppose $\rho, A, B$ are nonzero positive operators on $\caH$ with $\rho$ having full rank. Then 
		\begin{equation}
		\frac{\tr[\rho (A+B)^2]}{\tr(A+B)}\leq \frac{\tr(\rho A^2)}{\tr(A)}+\frac{\tr(\rho B^2)}{\tr(B)},
		\end{equation}
		and the inequality is saturated iff  $A$ is proportional to $B$.
		\begin{proof}According to the Cauchy-Schwarz inequality,
			\begin{equation}
			|\tr(\rho A B)|=|\tr(\rho^{1/2} A B\rho^{1/2})|\leq \sqrt{\tr(\rho A^2)\tr(\rho B^2)}. 
			\end{equation}
			The inequality is saturated iff $A\rho^{1/2}$ is proportional to $B\rho^{1/2}$, that is, $A$ is proportional to $B$. Therefore,
			\begin{align}
			\frac{\tr[\rho (A+B)^2]}{\tr(A+B)}&=
			\frac{\tr(\rho A^2)+\tr(\rho B^2)+\tr(\rho A B)+\tr(\rho BA)}{\tr(A+B)}\leq \frac{\bigl[\sqrt{\tr(\rho A^2)}+\sqrt{\tr(\rho B^2)}\,\bigr]^2} {\tr(A+B)}\nonumber\\
			&\leq  \frac{\tr(\rho A^2)}{\tr(A)}+\frac{\tr(\rho B^2)}{\tr(B)}.
			\end{align}
			If the first inequality is saturated, then  $A$ is proportional to $B$, in which case the second inequality is saturated automatically. 
		\end{proof}
	\end{lemma}

	\section{Tight coherent measurements} 
	
	In this section we prove \thsref{thm:UFS} and \ref{thm:UFSmin} in the main text, thereby clarifying the structure of tight coherent measurements. We also  provide more details on constructing minimal tight coherent measurements in dimension 3.

	\subsection{Proofs of \thsref{thm:UFS} and \ref{thm:UFSmin}}
	
	\Thref{thm:UFS} is an immediate consequence of the following theorem, which refines \thref{thm:UFS}. 
	\begin{theorem}\label{thm:UFSsupp}
		Let  $\{\Pi_\xi\}$ be a POVM  on $\caH^{\otimes 2}$ and $Q_\xi:=\tr_1(\Pi_\xi)+\tr_2(\Pi_\xi)$. Then the following statements are equivalent. 
		\begin{enumerate}
			\item $\{\Pi_\xi\}$ is tight coherent. 
			
			\item $\{Q_\xi\}$ is a generalized 2-design of purity $\frac{3d+1}{4d}$. 
			
			\item Each $\Pi_\xi $ is proportional to either the tensor power of a pure state or a Slater-determinant state, and $\{Q_\xi\}$ forms a generalized 2-design.
			
			\item $\{\Pi_\xi\}$ is a union of two POVMs $\{\Pi_\zeta^+ \}$ and $\{\Pi_\eta^- \}$ on $\caH_+$ and $\caH_-$, respectively; $\{Q_\zeta^+  \}$ forms a 2-design, and $\{Q_\eta^- \}$ forms a generalized 2-design of purity $\frac{1}{2}$.
		\end{enumerate}
	\end{theorem} 	
	\begin{remark}
		Here $Q_\zeta^+$ and $Q_\eta^-$ are defined in analogy to $Q_\xi$, that is, 
		$Q_\zeta^+=\tr_1(\Pi_\zeta^+)+\tr_2(\Pi_\zeta^+)$ and $Q_\eta^-=\tr_1(\Pi_\eta^-)+\tr_2(\Pi_\eta^-)$. If any of the four statements in \thref{thm:UFSsupp} holds, then $Q_\zeta^+$ are proportional to rank-1 projectors, and $Q_\eta^-$ are proportional to rank-2 projectors. When $d=2$, $Q_\eta^- $ are necessarily proportional to the identity, so the generalized 2-design $\{Q_\eta^- \}$ is trivial. 
	\end{remark}
	\begin{proof}
		Choose the affine parametrization specified in \eref{eq:AffinePara}, then 	
		the Fisher information matrix $I^{(2)}$ at the completely mixed state associated with $\{\Pi_\xi\}$ has matrix elements
		\begin{align}\label{eq:FisherMixTwo}
		I^{(2)}_{ab}=2\sum_{\xi}\frac{ \tr(E_a  Q_\xi) \tr(E_bQ_\xi)}{\tr (Q_\xi)},
		\end{align}
		Note that $\sum_\xi Q_\xi=2d$ and $\sum_\xi \tr(Q_\xi)=2d^2$, so $\{Q_\xi\}$ forms a POVM on $\caH$ up to scaling. Also, note the similarity between \eref{eq:FisherMixTwo} and  \eref{eq:FisherMixOne}. According  to \pref{pro:WFSmix} in the main text, $\{\Pi_\xi\}$ is weakly Fisher symmetric at the completely mixed state iff $\{Q_\xi\}$ is a generalized 2-design. In addition, 
		\begin{equation}
		\tr(I^{(2)})=2\sum_{\xi}\frac{d\tr(Q_\xi^2)-[\tr(Q_\xi)]^2}{d\tr(Q_\xi)}=4d^2\wp -4d, \quad \tr(J^{-1}I^{(2)})=\frac{1}{d}\tr(I^{(2)})=4d\wp -4.
		\end{equation}
		where $\wp:=\sum_\xi \tr( Q_\xi^2)/[2d^2\tr( Q_\xi)]$
		is the purity of $\{Q_\xi\}$. 
		
		By \thref{thm:GMT2}, if $\{\Pi_\xi\}$ is coherent, then $\tr(J^{-1}I^{(2)})=3d-3$, which implies that $\wp=(3d+1)/(4d)$. If  $\{\Pi_\xi\}$ is tight coherent, then it is also Fisher symmetric at the completely mixed state, so $\{Q_\xi\}$ is a generalized 2-design of purity $(3d+1)/(4d)$. Conversely, if $\{Q_\xi\}$ is a generalized 2-design of purity $(3d+1)/(4d)$, then $\tr(J^{-1}I^{(2)})=3d-3$, and $I$ is proportional to $J$, so $\{\Pi_\xi\}$ is Fisher symmetric. It follows that statements 1 and 2 are equivalent. 
		
		If statements 1 and 2 hold, then each $\Pi_\xi$ is proportional to either the tensor power of a pure state or  a Slater-determinant state
		according to  \thref{thm:GMT2}. In addition, $\{Q_\xi\}$ is a generalized 2-design. So
		statements 1 and 2 imply statement 3. On the other hand, $\{Q_\xi\}$ necessarily  has purity $(3d+1)/(4d)$ if each $\Pi_\xi$ is proportional to either the tensor power of a pure state or  a Slater-determinant state. So statement 3 implies statement 2. Consequently, statements 1, 2, and 3 are equivalent.

		It is easy to verify that statement 4 implies statement 2, which in turn implies statements 1 and 3 according to the above discussion. Now suppose  statement 3 holds. Then each $\Pi_\xi$ is proportional to either the tensor power of a pure state or  a Slater-determinant state, so $\{\Pi_\xi\}$ is a union of two POVMs $\{\Pi_\zeta^+ \}$ and $\{\Pi_\eta^- \}$ on $\caH_+$ and $\caH_-$, respectively, $\{Q_\zeta^+  \}$ is a 2-design, and $\{Q_\eta^- \}$ has purity $1/2$. 
		Given that both $\{Q_\xi \}$  and $\{Q_\zeta^+  \}$ are  generalized 2-designs, it follows that $\{Q_\eta^-\}$ is also a generalized 2-design.  Therefore, statement~3 implies statement~4. This observation completes the proof of \thref{thm:UFSsupp}.
	\end{proof}

	\begin{proof}[Proof of \thref{thm:UFSmin}]
		Suppose $\{\Pi_\xi\}$  is tight coherent. 
		According to \thref{thm:UFSsupp}, $\{\Pi_\xi\}$ is a union of two POVMs $\{\Pi_\zeta^+ \}$ and $\{\Pi_\eta^- \}$ on $\caH_+$ and $\caH_-$, respectively, $\{Q_\zeta^+  \}$ is a 2-design,  $\{Q_\eta^- \}$ is a generalized 2-design of 
		purity $1/2$, and $Q_\eta^-$ are proportional to rank-2 projectors. According to \pref{pro:gSICg2design}, both $\{Q_\zeta^+  \}$ and   $\{Q_\eta^- \}$ have at least $d^2$ elements, so $\{\Pi_\xi\}$ has at least $2d^2$ elements. If the lower bound is saturated, then  both $\{Q_\zeta^+  \}$ and   $\{Q_\eta^- \}$ have  $d^2$ elements, so  $\{Q_\zeta^+  \}$ forms a SIC by \crref{cor:SIC2design}, and   $\{Q_\eta^- \}$ forms a generalized SIC of purity $1/2$ by \pref{pro:gSICg2design}. The converse is an easy consequence of \pref{pro:gSICg2design} and \thref{thm:UFSsupp}. 
	\end{proof}

\subsection{Minimal tight coherent measurements in dimension~3}
According to \thref{thm:UFSmin}  in the main text, any tight coherent POVM on $\caH$ with $2d^2$ elements when $d\geq3$ is determined by a SIC and a generalized SIC composed of rank-2 projectors, and vice versa.  When $d=3$, if $\{B_\eta\}$ for $\eta=1,2,\ldots, 9$ is such a generalized SIC, then $\{1-B_\eta\}$ is a SIC \cite{Zaun11,ReneBSC04}. Therefore, minimal tight coherent POVMs in dimension~3 are in one-to-one correspondence with pairs of SICs. 

Since all SICs in dimension 3 are known \cite{Zaun11,ReneBSC04, Appl05, Zhu10, HughS16,Szol14}, all minimal tight coherent  POVMs in dimension~3 can be constructed explicitly. More precisely,  all SICs  in dimension 3 are covariant with respect to the Heisenberg-Weyl group, which is generated by the cyclic shift operator $X$  and the phase operator $Z:=\diag(1,\rme^{2\pi\rmi/3},\rme^{4\pi\rmi/3})$.  In addition, every SIC is equivalent to a SIC of the form $\{X^jZ^k|\psi(\phi)\>\}_{j,k=0,1,2}$ \cite{Zhu10,Zhu16Q}, where
\begin{equation}
|\psi(\phi)\rangle:=\frac{1}{\sqrt{2}}(0,1,-\rme^{\rmi \phi})^\rmT,\quad 0\leq \phi\leq \frac{\pi}{9}.
\end{equation}
Suppose $\{|\psi_\zeta\>\}$ and $\{|\varphi_\zeta\>\}$ are two  SICs in dimension 3. Let 
\begin{equation}
\Pi_\zeta^+ :=\frac{2}{3}(|\psi_\zeta\>\<\psi_\zeta|)^{\otimes 2}, \quad \Pi_\zeta^- :=\frac{1}{3}P_-(1-|\varphi_\zeta\>\<\varphi_\zeta|)^{\otimes 2}P_-,\quad \zeta=1,2,\ldots, 9.
\end{equation}
Then the union of $\{\Pi_\zeta^+ \}$ and $\{\Pi^-_\zeta\}$, that is, the POVM composed of all $\Pi_\zeta^+ $ and $\Pi^-_\zeta$,   is tight coherent. Conversely,   every minimal tight coherent POVM in dimension~3 has this form and thus can be constructed explicitly.

\end{document}